\documentclass[10pt]{IEEEtran}

\usepackage{amsmath, graphicx, amssymb, amsfonts, subfigure, setspace, amsthm, color}

\newcommand{\nchoosek}[2]{\left(\begin{array}{c}#1\\#2\end{array}\right)}
\newtheorem{dfn}{Definition}

\newtheorem{theorem}{Theorem}

\newtheorem{prop}{Proposition}
\newtheorem{lemma}{Lemma}

\begin{document}
\onecolumn
\doublespacing
\allowdisplaybreaks
\title{On the Delay Advantage of Coding\\in Packet Erasure Networks\thanks{
    The material of this paper was presented in part at the IEEE International Symposium on Information Theory 2009 and the IEEE Information Theory Workshop 2010.}}

\author{
\authorblockA{$\text{Theodoros K. Dikaliotis}^1$, $\text{Alexandros G. Dimakis}^2$, $\text{Tracey Ho}^1$, $\text{Michelle Effros}^1$\\
$\hspace{1mm}^1$California Institute of
Technology $\hspace{1mm}^2$University of Southern California\\
{$^1$\{tdikal, tho, effros\}@caltech.edu
$^2$dimakis@usc.edu}}}

\newcommand{\beq}{\begin{equation}}
\newcommand{\eeq}{\end{equation}}
\newcommand{\bea}{\begin{eqnarray}}
\newcommand{\eea}{\end{eqnarray}}

\newcommand{\stexp}{\mbox{$\mathbb{E}$}}    

\newcommand{\Prob}{\ensuremath{\mathbb{P}}}
\newcommand{\var}{\mbox{$\mathsf{Var}$}}    
\def\qed{\quad \vrule height6.5pt width6pt depth0pt} 
\newcommand{\bigO}{\mbox{$\mathsf{O}$}}  
\newcommand{\X} {\mathcal{X}}
\def\qed{\quad \vrule height6.5pt width6pt depth0pt} 
\long\def\symbolfootnote[#1]#2{\begingroup%
\def\thefootnote{\fnsymbol{footnote}}\footnote[#1]{#2}\endgroup}

\maketitle


\begin{abstract}
We consider the delay of network coding compared to routing with retransmissions in packet erasure networks with probabilistic erasures. We investigate the sub-linear term in the block delay required for unicasting $n$ packets and show that there is an unbounded gap between network coding and routing. In particular, we show that delay benefit of network coding scales at least as $\sqrt{n}$. Our analysis of the delay function for the routing strategy involves a major technical challenge of computing the expectation of the maximum of two negative binomial random variables. This problem has been studied previously and we derive the first exact characterization which may be of independent interest.
We also use a martingale bounded differences argument to show that the actual coding delay is tightly concentrated around its expectation. \end{abstract}

\begin{keywords}\vspace{-4mm}
Block delay, network coding, packet erasure correction, retransmission, unicast.

\end{keywords}

\section{Introduction}
\label{sec:intro}





This paper considers the block delay for unicasting a file consisting of $n$ packets over a packet erasure network with probabilistic erasures. Such networks have been extensively studied from the standpoint of capacity. Various schemes involving coding or retransmissions have been shown to be capacity-achieving for unicasting in networks with packet erasures, e.g.~\cite{dana06capacity, lun04coding, smith08wireless, neely09optimal}. For a capacity-achieving strategy, the expected block delay for transmitting $n$ packets is $\frac{n}{C}+D(n)$ where $C$ is the minimum cut capacity and the delay function $D(n)$ is sublinear in $n$ but differs in general for different strategies. In general networks, the optimal $D(n)$ is achieved by random linear network coding, in that decoding succeeds with high probability for any realization of packet erasure events for which the corresponding minimum cut capacity is $n$\footnote{The field size and packet length are assumed in this paper to be sufficiently large so that  the probability of rank-deficient choices of coding coefficients can be neglected, along with the fractional overhead of specifying the random coding vectors.}. However, relatively little has been known previously about the behavior of the delay function $D(n)$ for coding or retransmission strategies.

In this paper, we analyze the delay function $D(n)$ for random linear network coding (coding for short) as well as an uncoded hop-by-hop retransmission strategy (routing for short) where only one copy of each packet is kept in intermediate node buffers. Schemes  such as~\cite{roofnet:exor-hotnets03,neely09optimal} ensure that there is only one copy of each packet in the network; without substantial non-local coordination or feedback, it is complicated for an uncoded topology-independent scheme to keep track of multiple copies of packets at intermediate nodes and prevent capacity loss from duplicate packet transmissions. We also assume that the routing strategy fixes how many packets will traverse each route a priori based on link statistics, without adjusting to link erasure realizations. While routing strategies could dynamically re-route packets under atypical realizations, this would not be practical if the min-cut links are far from the source. On the other hand, network coding allows redundant packets to be transmitted efficiently in a topology-independent manner, without feedback or coordination, except for an acknowledgment from the destination when it has received the entire file. As such, network coding can fully exploit variable link realizations. These differences result in a coding advantage in delay function $D(n)$ which, as we will show, can be unbounded with increasing $n$.

A major technical challenge in the analysis of the delay function for the routing strategy involves computing the expectation of the maximum of two independent negative binomial random variables. This problem has been previously studied in~\cite{maximum_statistics}, where authors explain in detail why it is complicated\footnote{Authors in~\cite{maximum_statistics} deal with the expected maximum of any number of negative binomial distributions but the difficulty remains even for two negative binomial distributions.} and derive an approximate solution to the problem. Our analysis addresses this open problem by finding an exact expression and showing that it grows to infinity at least as the square root of $n$.


Related work on queuing delay in uncoded~\cite{rubin74communication,shalmon87exact} and coded~\cite{ephremides2006} systems has considered the case of random arrivals and their results pertain to the delay of individual packets in steady state. This differs from our work which considers the delay for communicating a fixed size batch of $n$ packets that are initially present at the source.

\subsection{Main results}

For a line network, the capacity is given by the worst link. We show a finite bound on the delay function that applies to both coding and the routing scheme when there is a single worst link.
%

%
\begin{theorem}
Consider $n$ packets communicated through a line network of $\ell$ links with erasure probabilities $p_1,p_2,\ldots, p_\ell$ where there is a unique worst link:
\begin{equation*}
p_m := \displaystyle\max_{1\leq i\leq \ell} p_i, \quad p_i <p_m<1 \quad \forall \, i \neq m.
\end{equation*}
The expected time $\stexp T_n$ to send all $n$ packets either with coding or routing is:
\begin{align}
\stexp T_n  =  \frac{n}{1-\displaystyle\max_{1\leq i\leq \ell} p_i}+ D(n,p_1,p_2,\ldots, p_\ell),
\label{delay_fun_1}
\end{align}
where the delay function $D(n,p_1,p_2,\ldots, p_\ell)$ is non-decreasing in $n$ and upper bounded by:
\begin{equation*}
\bar D(p_1,p_2, \ldots, p_\ell) :=
\sum_{i=1, i \neq m}^{\ell}\frac{p_m}{p_m - p_i}.
\end{equation*}
\label{thm:Expected_values_multi_hop_line_network}
\end{theorem}
If on the other hand there are two links that take the worst value,
then the delay function is not bounded but still exhibits sublinear behavior. Pakzad et al.~\cite{pakzad05coding} show that
in the case of a  line  network with identical links, the optimal delay function grows as $\sqrt{n}$.  This is achieved by both coding and the routing strategy\footnote{The result in~\cite{pakzad05coding} is derived for the routing strategy which is delay-optimal in a line network; as discussed above,
coding in a sufficiently large field is delay-optimal in any network.}.

In contrast, for parallel path networks, we show that the delay function behaves quite differently for the coded and uncoded schemes.
\begin{theorem}
The expected time $\stexp T_n^\text{c}$ taken to send $n$ packets using coding over a $k$-parallel path multi-hop network is
\begin{align*}
\stexp T_{n}^\text{c} = \frac{n}{k-\displaystyle\sum_{i=1}^k\mathop{\max}_{1\leq j\leq\ell} p_{ij}}+D^\text{c}_n
\end{align*}
where the delay function $D^\text{c}_n$ depends on all the erasure probabilities $p_{ij}$, for $i\in\{1,\ldots,k\}$, $1\leq j\leq \ell$. In the case where there is single worst link in each path $D^\text{c}_n$ is bounded, \emph{i.e.} $D^\text{c}_n\in\bigO(1)$ whereas if there are multiple worst links in at least one path then $D^\text{c}_n\in\bigO(\sqrt{n})$. The result holds regardless of any statistical dependence between erasure processes on different paths.
\label{thm:coding}
\end{theorem}

\begin{theorem}
The expected time $\stexp T_n^\text{r}$ taken to send  $n$ packets through a $k$-parallel path network by routing is
\begin{align}
\stexp T_{n}^\text{r} = \frac{n}{k-\displaystyle\sum_{i=1}^k\mathop{\max}_{1\leq j\leq\ell} p_{ij}}+D^\text{r}_n
\label{eqn:finish}
\end{align}
where the delay function $D^\text{r}_n$ depends on all the erasure probabilities $p_{ij}$, for $i\in\{1,\ldots,k\}$, $1\leq j\leq \ell$ and grows at least as $\sqrt{n}$, \emph{i.e.} $D^\text{r}_n\in\Omega(\sqrt{n})$.
\label{thm:routing}
\end{theorem}

The above results on parallel path networks  generalize to arbitrary topologies. We define {\it single-bottleneck networks} as networks that have a single min-cut.

\begin{theorem}
In a  network of erasure channels with a single source $\mathcal{S}$ and a single receiver $\mathcal{T}$ the expected time $\stexp T_{n}^\text{r}$ taken to send $n$ packets by routing is
\begin{align*}
\stexp T_{n}^\text{r} = \frac{n}{C}+\hat{D}^\text{r}_n
\end{align*}
where $C$ is the capacity of the network and $\hat{D}^\text{r}_n\in\Omega(\sqrt{n})$. In the case of network coding the expected time $\stexp T_{n}^\text{c}$ taken to send $n$ packets is
\begin{align*}
\stexp T_{n}^\text{r} = \frac{n}{C}+\hat{D}^\text{r}_n
\end{align*}
where  $\hat{D}^\text{c}_n\in\bigO(1)$ for single-bottleneck networks.
\label{thm:super_general}
\end{theorem}

We also prove the following concentration result:
\begin{theorem}
The time $T_n^\text{c}$ for $n$ packets to be transmitted from a source to a sink over a  network of erasure channels using network coding is concentrated around its expected value with high probability. In particular for sufficiently large $n$:
\begin{align}
\Prob\left[\left|T^{\text{c}}_n-\stexp T^{\text{c}}_n\right|> \epsilon_{n}\right] \leq \frac{2C}{n}+  o\left(\frac{1}{n}\right),
\label{eqn:concentration}
\end{align}
where $C$ is the capacity of the network and $\epsilon_n$ represents the corresponding deviation and is equal to $\epsilon_{n} = n^{1/2+\delta}\slash C$, $\delta\in (0,1/2)$.
\label{thm:concentration}
\end{theorem}
Since $\stexp T_n^\text{c}$ grows linearly in $n$ and the deviations $\epsilon_n$ are sublinear, $T_n^\text{c}$ is tightly concentrated around its expectation for large $n$ with probability approaching one. Subsequent to our initial conference publications~\cite{ted2009,dikaliotis10delay}, further  results on delay for line networks have been obtained by~\cite{heidarzadeh12coding,heidarzadeh12fast}.

\section{Model}
\label{sec:Model}

We consider a  network $\mathcal{G}=(\mathcal{V},\mathcal{E})$ where $\mathcal{V}$ denotes the set of nodes and $\mathcal{E}=\mathcal{V}\times\mathcal{V}$ denotes the set of edges or links. We assume a discrete time model, where at each time step each node $v\in\mathcal{V}$ can transmit one packet on its outgoing edges. For every edge $e\in\mathcal{E}$ each transmission succeeds with probability $1-p_{e}$ or the transmitted packet gets erased with probability $p_{e}$; erasures across different edges and time steps are assumed to be independent. In our model, in case of a success the packet is assumed to be transmitted to the next node instantaneously, i.e. we ignore the transmission delay along the links. We assume that no edge fails with probability 1 (i.e.~$p_e<1$ for all $e\in\mathcal{E}$) since in such a case we can remove that edge from the network.

Within network $\mathcal{G}$ there is a single source $\mathcal{S}\in\mathcal{V}$ that wishes to transmit $n$ packets to a single destination $\mathcal{T}$ in $\mathcal{G}$. We investigate the expected time it takes for the $n$ packets to be received by $\mathcal{T}$ under two transmission schemes, network coding and routing. When network coding is employed, each packet transmitted by a node $v\in\mathcal{V}$ is a random linear combination of all previously received packets at the node $v$. The destination node $\mathcal{T}$ decodes once it has received $n$ linearly independent combinations of the initial packets. When routing is employed, the number of packets transmitted in each path is fixed ahead of the transmission, in such a way that the expected time for all $n$ packets to reach destination $\mathcal{T}$ is minimized.

All nodes in the network are assumed to have sufficiently large buffers to store the necessary number of packets to accommodate the transmission scheme. In the case of routing, we assume an automatic repeat request (ARQ) scheme with instantaneous feedback available on each hop.  Thus, a node can drop a packet that has been successfully received by the next node. For the case of coding, as explained in~\cite{jay2007}, information travels through the network in the form of innovative packets, where a packet is innovative for a node $v$ if it is not in the linear span of packets previously received by $v$. For simplicity of analysis, we assume that a node can store up to $n$ linearly independent packets; smaller buffers can be used in practice\footnote{By the results of~\cite{haeulper11optimality}, the buffer size needed for coding is no larger than that needed for routing.}.
Feedback is not needed except when the destination $\mathcal{T}$ receives all the information and signals the end of transmission to all nodes. Our results hold without any restrictions on the number of packets $n$ or the number of edges in the network, and there is no requirement for the network to reach steady state.

\section{Line Networks}
\label{sec:Line_networks}

\begin{figure}
\begin{center}
\includegraphics[clip=true, trim=0mm 0mm 0mm 0mm, width=0.95\columnwidth]{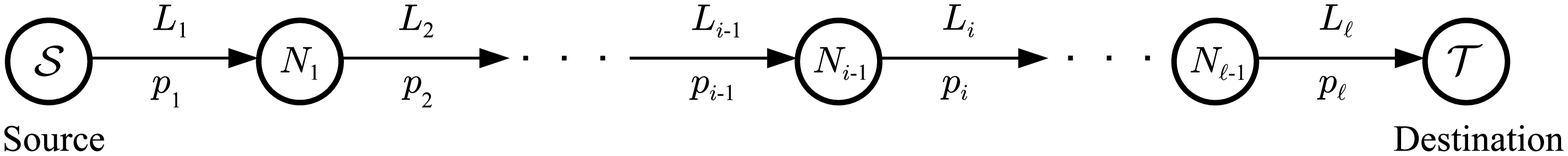}
\end{center}
\caption{Multi-hop line network}
\label{fig:Multi_hop_network}
\end{figure}

The line network under consideration is depicted in Figure~\ref{fig:Multi_hop_network}. The network consists of $\ell$ links $L_i$, $1\leq i\leq\ell$ and $\ell+1$ nodes $N_j$, $0\leq j\leq\ell$. Node $N_j$, $0\leq j\leq\ell-1$ is connected to node $N_{j+1}$ to its right through the erasure link $L_{j+1}$, where we assume that the source $\mathcal{S}$ and the destination $\mathcal{T}$ are also defined as nodes $N_0$ and $N_\ell$ respectively. The probability of transmission failure on each link $L_i$ is denoted by $p_{i}$.

For the case of a line network there is no difference between network coding and routing in the expected time it takes to transmit a fixed number of packets. Note that coding at each hop (network coding) is needed to achieve minimum delay in the absence of feedback, whereas coding only at the source is suboptimal in terms of throughput  and delay~\cite{lun04coding}. 

\begin{IEEEproof}[Proof of Theorem~\ref{thm:Expected_values_multi_hop_line_network}]
By using the interchangeability result on service station from Weber~\cite{weber92interchangeability}, we can interchange the position of any two links without affecting the departure process of node $N_{\ell-1}$ and therefore the delay function. Consequently, we can interchange the worst link in the queue (which is unique from the assumptions of Theorem~\ref{thm:Expected_values_multi_hop_line_network}) with the first link, and thus we will assume that the first link is the worst link ($p_2, p_3,\ldots, p_\ell < p_1 < 1$).

Note that in a line network, under coding  the  subspace spanned by all  packets received so far at a node $N_i$  contains that of its next hop node $N_{i+1}$, similarly to the case of routing where the set of packets received at a node $N_i$  is a superset of that of its next hop node $N_{i+1}$. Let the random variable $R_i^n, 1\leq i \leq \ell-1,$ denote the rank difference between node $N_i$ and node $N_{i+1}$, at the  moment packet $n$ arrives at $N_1$. This is exactly the number of packets present at node $N_i$ that are innovative for $N_{i+1}$ (which for brevity we refer to simply as innovative packets at node $N_i$ in this proof) at the random time when packet $n$ arrives at $N_1$. For any realization of erasures, the evolution of the number of innovative packets at each node is the same under coding and routing.

The time $T_n$ taken to send $n$ packets from the source node $\mathcal{S}$ to the destination $\mathcal{T}$ can be expressed as the sum of time $T_n^{(1)}$ required for all the $n$ packets to cross the first link and the time $\tau_n$ required for all the remaining innovative packets $R_1^n,\ldots,R_{\ell-1}^n$ at nodes $N_1,\ldots,N_{\ell-1}$ respectively to reach the destination node $\mathcal{T}$:
\begin{align*}
T_n = T_n^{(1)}+\tau_n.
\end{align*}
All the quantities in the equation above are random variables and we want to compute their expected values. Due to the linearity of the expectation
\begin{align}
\stexp T_n = \stexp T_n^{(1)} + \stexp \tau_n
\label{eqn:Exp_sum_of_times}
\end{align}
and by defining $X_j^{(1)}, 1 \leq j \leq n$ to be the time taken for packet $j$ to cross the first link, we get:
\begin{eqnarray}
\stexp T_n^{(1)} =
\sum_{j=1}^{n}\stexp X_j^{(1)} = \frac{n}{1-p_1}
\label{eqn_claim:Expected_values_1}
\end{eqnarray}
since $X_j^{(1)}, 1 \leq j \leq n,$ are all geometric random variables ($\Prob \left(X_j^{(1)}=k\right)=(1-p_1)\cdot p_1^{k-1}, k\geq 1$). Therefore combining equations (\ref{eqn:Exp_sum_of_times}) and (\ref{eqn_claim:Expected_values_1}) we get:
\begin{eqnarray}
\stexp T_n^{(1)} = \frac{n}{1-p_1} + \stexp \tau_n.
\label{eqn_claim:Expected_values_2}
\end{eqnarray}

Equations (\ref{delay_fun_1}), (\ref{eqn_claim:Expected_values_2}) give
\begin{equation*}
D(n,p_1,p_2,\ldots,p_\ell)=
\stexp \tau_n
\label{eqn:Delay_function}
\end{equation*}
%
which is the expected time taken for all the remaining innovative packets at nodes $N_1,\ldots,N_{\ell-1}$ to reach the destination. For the simplest case of a two-hop network ($\ell=2$) we can derive recursive formulas for computing this expectation for each $n$. Table \ref{table:E(Rem_n)} has closed-form expressions for the delay function $D(n, p_1, p_2)$ for $n=1,\ldots,4$.
\begin{table*}
\centering
\label{table:E(Rem_n)}
\caption{The delay function $D(n, p_1, p_2)$ for different values of $n$}
\begin{tabular}{c c}
\hline\hline
$n$ & $D(n, p_1, p_2)$\\ \hline
1 & $\frac{1}{1-p_2}$ \\ [1.0ex]
2 & $\frac{2}{1-p_2}-\frac{1}{1-p_1 p_2}$  \\  [1.0ex]
3 & $\frac{1+p_2 \left(2-p_1 \left(6-p_1+(2-5 p_1) p_2+(1-3 (1-p_1) p_1) p_2^2\right)\right)}{(1-p_2)(1-p_1 p_2)^3}$  \\ [1.0ex]
4 & $\frac{\tiny\left\{\begin{array}{c}
1+p_2 (3-p_1 (11+4 p_1^4 p_2^4+p_2 (5+(5-p_2) p_2)+p_1^3 p_2 (1-p_2 (5+2 p_2 (5+3 p_2)))\\
-p_1 (4+p_2 (15+p_2 (21-(1-p_2) p_2)))+p_1^2 (1-p_2 (1-p_2 (31+p_2 (5+4 p_2))))))
\end{array}\right\}}
{(1-p_2)(1-p_1 p_2)^5}$ \\ [1.0ex]
\hline
\end{tabular}
\end{table*}
It is seen that as $n$ grows, the number of terms in the above expression increases rapidly, making these exact formulas impractical, and as expected for larger values of $\ell$ ($\geq 3$) the situation only worsens. Our subsequent analysis derives tight upper bounds on the delay function $D(n,p_1,p_2,\ldots,p_\ell)$ for any $\ell$ which do not depend on $n$.

The $(\ell-1)$-tuple $Y_n = (R_1^n,\ldots,R_{\ell-1}^n)$ representing the number of innovative packets remaining at nodes $N_1,\ldots,N_{\ell-1}$ at the moment packet $n$ arrives at node $N_{1}$ (including packet $n$) is a multidimensional Markov process with state space $E \subset \mathbb{N}\hspace{0.8mm}^{\ell-1}$ (the state space is a proper subset of $\mathbb{N}\hspace{0.8mm}^{\ell-1}$ since $Y_n$ can never take the values $(0,*,\ldots,*)$ where the $*$ represents any integer value). Using the coupling method~\cite{stock2} and an argument similar to the one given at Proposition 2 in~\cite{aggregate} it can be shown that $Y_n$ is a stochastically increasing function of $n$ (meaning that as $n$ increases there is a higher probability of having more innovative packets at nodes $N_1,\ldots,N_{\ell-1}$).

\begin{prop}
The Markov process $Y_n = (R_1^n,\ldots,R_{\ell-1}^n)$ is $\preceq_{\text{st}}$-increasing.
\label{prop:The_markov_chain}
\end{prop}
\begin{proof}
Given in Appendix~\ref{appen:Proof_of_stochastically_increasing} along with the necessary definitions.
\end{proof}

A direct result of Proposition~\ref{prop:The_markov_chain} is that the expected time taken $\stexp \tau_n$ for the remaining innovative packets at nodes $N_1,\ldots,N_{\ell-1}$ to reach the destination is a non-decreasing function of  $n$:
\begin{equation}
\stexp \tau_n \leq \stexp \tau_{n+1} \leq \displaystyle\lim_{n\rightarrow\infty} \stexp \tau_n
\label{eqn:inequality_of_expectations}
\end{equation}
where the second inequality is meaningful when the limit exists.

Innovative packets travelling in the network from node $N_1$ to the destination node $\mathcal{T}$ can be viewed as customers travelling through a network of service stations in tandem. Indeed, each innovative packet (customer) arrives at the first station (node $N_1$) with a geometric arrival process and the transmission (service) time is also geometrically distributed. Once an innovative packet has been transmitted (serviced) it leaves the current node (station) and arrives at the next node (station) waiting for its next transmission (service).

It is helpful to assume the first link to be the worst one in order to use the results of Hsu and Burke in~\cite{hsu76behavior}. The authors proved that a tandem network with geometrically distributed service times and a geometric input process, reaches steady state as long as the input process is slower than any of the service times. Our line network is depicted in Figure~\ref{fig:Multi_hop_network} and the input process (of innovative packets) is the geometric arrival process at node $N_1$ from the source $\mathcal{S}$. Since $p_2, p_3,\ldots, p_\ell < p_1$ the arrival process is slower than any service process (transmission of the innovative packet to the next hop) and therefore the network in Figure~\ref{fig:Multi_hop_network} reaches steady state.

Sending an arbitrarily large number of packets ($n\rightarrow \infty$) makes the problem of estimating  $\displaystyle\lim_{n\rightarrow\infty}\stexp \tau_n$\footnote{If the network was not reaching a steady state the above limit would diverge.} the same as calculating the expected time taken to send all the remaining innovative packets at nodes $N_1,\ldots,N_{\ell-1}$ to reach the destination $\mathcal{T}$ at steady state. This is exactly the expected end-to-end delay for a single customer in a line network that has reached equilibrium. This quantity has been calculated in~\cite{daduna01queueing} (page 67, Theorem 4.10) and is equal to
\begin{equation}
\lim_{n\rightarrow\infty} \stexp \tau_n =\sum_{i=2}^{\ell}\frac{p_1}{p_1-p_i}.
\label{eqn:the_expression_of_the_limit}
\end{equation}
Combining equations (\ref{eqn:inequality_of_expectations}) and (\ref{eqn:the_expression_of_the_limit}) and changing $p_1$ to $p_m:=\max p_i < 1$
concludes the proof of Theorem~\ref{thm:Expected_values_multi_hop_line_network}.
\end{IEEEproof}

\section{$k$-parallel Path Network}
\label{sec:The_two_parallel_paths}

We define the \emph{$k$-parallel path network} as the network depicted in Figure~\ref{fig:parallel_multihop_network}. This network consists of $k$ parallel multi-hop line networks  (paths) with $k\ell$ nodes and $k\ell$ links, with $\ell$ links in each path (our results are readily extended to networks with different number of links in each path). Each node $N_{i(j-1)}$ is connected to the node $N_{ij}$ on its right by a link $L_{ij}$, for $i\in\{1,\ldots,k\}$ and $1\leq j\leq \ell$ where for consistency we assume that the source $\mathcal{S}$ and the destination $\mathcal{T}$ are defined as nodes $N_{i0}$ and $N_{i\ell}$, $i\in\{1,\ldots,k\}$, respectively.

\begin{figure}
\begin{center}
\includegraphics[width=1.0\columnwidth]{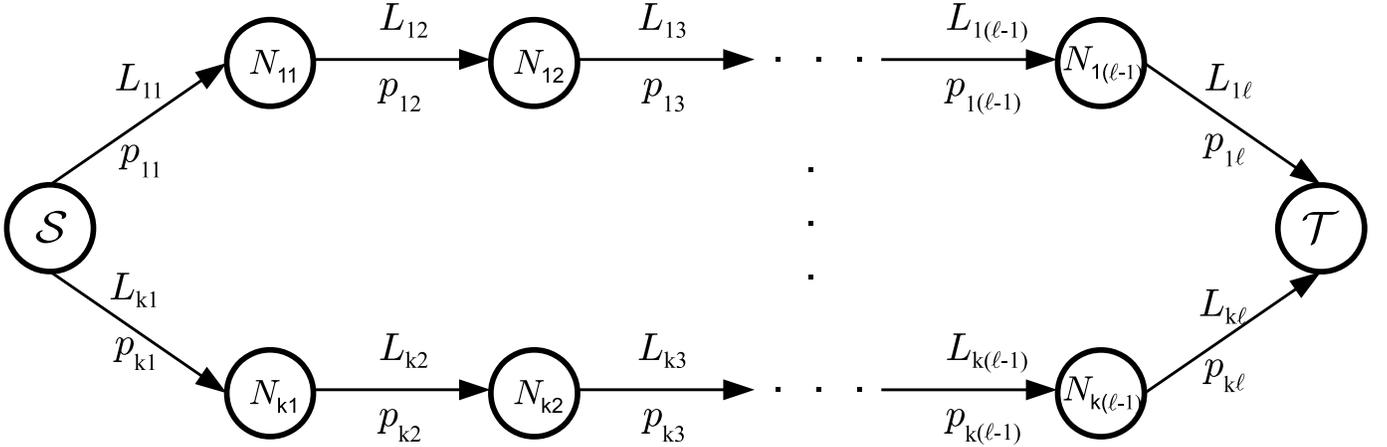}
\end{center}
\caption{Two parallel multi-hop line networks having links with different erasure probabilities}
\label{fig:parallel_multihop_network}
\end{figure}



For the case of routing with retransmissions, the source $\mathcal{S}$ divides the $n$ packets between the different paths so that the time taken to send all the packets is minimized in expectation. This is accomplished by having the number of packets that cross each path to be proportional to the capacity of the path. Indeed, if the source $\mathcal{S}$ sends $n_1,\ldots,n_k$ number of packets though each path then according to Theorem~\ref{thm:Expected_values_multi_hop_line_network} the expected time to send these packets is $\frac{n_i}{1-p_{1i}}+D_{n_i}$, $i\in\{1,\ldots,k\}$, where $D_{n_i}$ are bounded delay functions. The values $n_i$ are chosen so that the linear terms of the above expected values are equal, \emph{i.e.} $\frac{n_1}{1-p_{11}}=\ldots=\frac{n_k}{1-p_{k1}}$ and $n_1+\ldots+n_k=n$. Therefore the choice of
\begin{align}
n_i = \frac{n(1-p_{i1})}{k-\displaystyle\sum_{i=1}^k p_{i1}},\ i\in\{1,\ldots,k\}
\label{eqn:minimize_routing_time}
\end{align}
minimizes the expected time to send the $n$ packets. Therefore from now on, when routing is performed, source $\mathcal{S}$ is assumed to send $n(1-p_{i1})\slash(k-\sum_{i=1}^k p_{i1})$ over each path $i$.\footnote{To simplify the notation we will assume that all numbers $n(1-p_{i1})\slash(k-\sum_{i=1}^k p_{i1})$ are integers. Our results extend to the case that those numbers are not integers by rounding them to the closest integer.}

\subsection{Coding Strategy}
\begin{figure}[ht]
\begin{center}
\includegraphics[width=0.6\columnwidth]{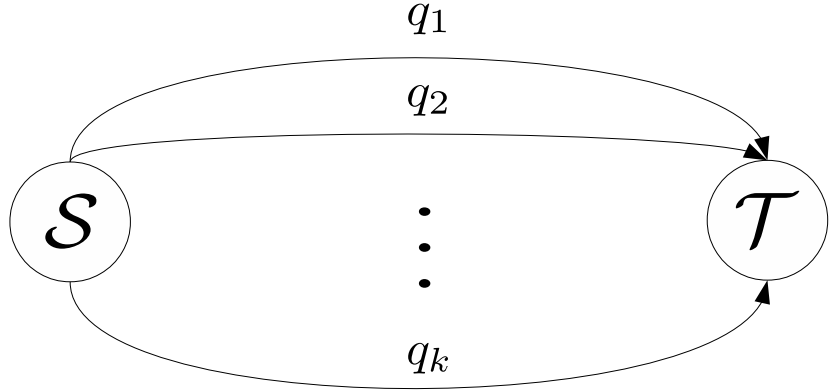}
\end{center}
\caption{A network of $k$ parallel erasure links with erasure probabilities $q_1,\ldots,q_k$ connecting source $\mathcal{S}$ and destination $\mathcal{T}$.}
\label{fig:parallelqueues}
\end{figure}

Before we analyze the expected time $\stexp T_{n}^\text{c}$ taken to send $n$ packets through the network in Figure~\ref{fig:parallel_multihop_network} using coding (where the $c$ superscript stands for coding), we prove the following proposition that holds for the simplified network of $k$ parallel erasure links connecting the source to the destination as in Figure~\ref{fig:parallelqueues}.
\begin{prop}
The expected time $\stexp \hat{T}_n^\text{c}$ taken to send by coding $n$ packets from source $\mathcal{S}$ to destination $\mathcal{T}$ through $k$ parallel erasure links with erasure probabilities  $q_1,\ldots,q_{k}$ respectively is
\begin{align*}
\stexp \hat{T}_n^\text{c} = \frac{n}{k-\sum_{i=1}^{k} q_i}+B_{n}
\end{align*}
where $B_{n}$ is a bounded term. This relation holds regardless of any statistical dependence between the erasure processes on different links.
\label{prop:two_parallel_links_coding}
\end{prop}
\begin{IEEEproof}
We define $A_0,A_1,\ldots,A_k$ to be the probabilities of having $0,1,\ldots,k$ links  succeed at a specific time instance. The recursive formula for $\stexp \hat{T}^c_n$ is:
\begin{eqnarray}
\stexp \hat{T}_n &=& A_0\cdot(\stexp\hat{T}^c_n+1) + A_1\cdot(\stexp \hat{T}^c_{n-1}+1)+\ldots+ A_k\cdot(\stexp\hat{T}^c_{n-k}+1)\notag\\
\Leftrightarrow (1-A_0)\cdot\stexp \hat{T}_n &=& A_1\cdot\stexp \hat{T}_{n-1}+\ldots+ A_k\cdot\stexp \hat{T}_{n-k}+1
\label{eqn:recurrence_general}
\end{eqnarray}
where $\stexp \hat{T}_m=0$ for $m\leq 0$ and the last term in (\ref{eqn:recurrence_general}) is obtained from the relation $\sum_{i=0}^{k}A_i=1$.

The general solution of (\ref{eqn:recurrence_general}) is given by the sum of a homogeneous solution and a special solution. A special solution for the non-homogeneous recursive equation (\ref{eqn:recurrence_general}) is linear $D\cdot n$ where after some algebra $D =1\slash(A_1+2 A_2+\ldots+k A_k)$, which is the inverse of the expected number of links succeeding in a given instant.  Therefore $D=1\slash(k-\sum_{i=1}^{k}q_i)$, independent of any statistical dependence between erasures on different links.

The homogeneous solution of linear recurrence relation with constant coefficients (\ref{eqn:recurrence_general}) can be expressed in terms of the roots of the characteristic equation $p(x)=(1-A_0) x^k-A_1 x^{k-1}-\ldots-A_k$~\cite[Section $3.2$]{recursion:bk}. We will prove that the characteristic equation has $x = 1$ as a root and all the other roots have absolute value less than $1$. Indeed since $A_0+\ldots+A_k=1\Rightarrow(1-A_0)-A_1-\ldots-A_k=0$, therefore $x=1$ is a root of $p(x)$; now assume that $x=1$ is a multiple root of $p(x)$. Then
\begin{eqnarray*}
p'(1)=0&\Leftrightarrow&k(1-A_0)-(k-1) A_1-\ldots-A_{k-1}=0\\
&\Leftrightarrow& k(1-A_0)-(k-1) A_1-\ldots-(k-(k-1)) A_{k-1}=0\\
&\Leftrightarrow& k=k(A_0+A_1+\ldots+A_{k-1})-A_1-2 A_2-\ldots-(k-1)A_{k-1}\\
&\Leftrightarrow& k=k(1-A_k)-A_1-2 A_2-\ldots-(k-1) A_{k-1}\\
&\Leftrightarrow& k=k-(A_1+2 A_2+\ldots+k A_k)\\
&\Leftrightarrow& 0=A_1+2 A_2+\ldots+k A_k\\
&\Leftrightarrow& k=p_1+p_2+\ldots+p_k
\end{eqnarray*}
This implies that all links fail with probability $1$, which contradicts the assumption from Section~\ref{sec:Model} that no link fails with probability $1$. Assume now that characteristic equation $p(x)$ has a complex root $x=r\cdot e^{i\cdot\phi}$ where $|x|>1$ or equivalently $r>1$. Define $f(x)=x^k$ and $g(x)=A_0x^k+A_1 x^{k-1}+\ldots+A_k$ then $p(x)=0$ is equivalent to $f(x)=g(x)$ but this last equality cannot hold since $|g(x)|<|f(x)|$ for $|x|>1$. Indeed $|g(x)|\leq A_0|x|^k+A_1|x|^{k-1}+\ldots+A_k=A_0 r^k+A_1 r^{k-1}+ \ldots+A_k<(A_0+A_1+\ldots+A_k) r^k=r^k =|f(x)|$.

Let $R=\big\{r:p(r)=0\big\}$ be the set of all roots of $p(x)$. The general solution for recursion formula (\ref{eqn:recurrence_general}) is
\begin{align*}
\stexp \hat{T}^c_n = \frac{n}{k-\sum_{i=1}^{k}q_i} + \sum_{r_j\in R} F_j r_j^n cos(n\cdot\phi_j) + \sum_{r_j\in R} G_j r_j^n sin(n\cdot\phi_j).
\end{align*}
We can set
\begin{align}
B_n = \sum_{r_j\in R} F_j r_j^n cos(n\cdot\phi_j) + \sum_{r_j\in R} G_j r_j^n sin(n\cdot\phi_j)
\end{align}
and since $|B_n|\leq\displaystyle\mathop{\sum}_{r_j\in R}|F_j|+|G_j|$ this concludes our proof.
\end{IEEEproof}

Now we are ready to prove the following theorem for the $k$-parallel path network shown in Figure~\ref{fig:parallel_multihop_network}.
\begin{IEEEproof}[Proof of Theorem~\ref{thm:coding}]
As discussed in the proof of Theorem~\ref{thm:Expected_values_multi_hop_line_network}, by using the results of~\cite{weber92interchangeability} we can interchange the position of the first link of each path with one of the worst links of the path without affecting the arrival process at the receiver $\mathcal{T}$. Therefore without loss of generality we will assume that the first link in each path is one of the worst links in the path. Also, as in the proof of Theorem~\ref{thm:Expected_values_multi_hop_line_network}, for brevity we refer to packets present at a node $N_i$ that are innovative for the next hop node $N_{i+1}$ as innovative packets at node $N_i$.

The time $T_{n}^\text{c}$ taken to send $n$ packets from source $\mathcal{S}$ to the destination $\mathcal{T}$ in Figure~\ref{fig:parallel_multihop_network} can be expressed as the sum of the time $\hat{T}_{n}^\text{c}$ required for all $n$ packets to reach one of nodes $N_{11},\ldots,N_{k1}$ and the remaining time $\tilde{T}_{n}^\text{c}$ required for all innovative packets remaining in the network to reach the destination $\mathcal{T}$, \emph{i.e.}
\begin{align}
T_{n}^\text{c} = \hat{T}_{n}^\text{c}+\tilde{T}_{n}^\text{c}.
\label{eqn:proof_sum_of_coding_times}
\end{align}
As in the proof of Theorem~\ref{thm:Expected_values_multi_hop_line_network} all quantities in equation (\ref{eqn:proof_sum_of_coding_times}) are random variables and we want to compute their expected values. Due to linearity of expectation,
\begin{align}
\stexp T_{n}^\text{c} = \stexp\hat{T}_{n}^\text{c}+\stexp\tilde{T}_{n}^\text{c},
\label{eqn:proof_sum_of_expected_coding_times}
\end{align}
where by Proposition~\ref{prop:two_parallel_links_coding},
\begin{align}
\stexp\hat{T}_{n}^\text{c}=\frac{n}{k-\displaystyle\sum_{i=1}^k p_{i1}}+B_{n}
\label{eqn:proof_expected_first_term}
\end{align}
where $B_n$ is bounded. This holds regardless of any statistical dependence between the erasure processes on the first link of each path, and the remainder of the proof is unaffected by any statistical dependence between erasure processes on different paths.

The time $\stexp\tilde{T}_n^\text{c}$ required to send all the remaining innovative packets at nodes $N_{ij}$ ($i\in\{1,\ldots,k\}$, $j\in\{2,\ldots,\ell-1\}$) to the destination is less than the expected time $\stexp\tilde{\tau}$ it would have taken if all the remaining innovative packets were returned back to the source $\mathcal{S}$ and sent to the destination $\mathcal{T}$ using only the first path. Let $R_{ij}$ denote the number of remaining innovative packets at node $N_{ij}$ at the moment the $n^\text{th}$ packet has arrived at one of the $k$ nodes $N_{11},\ldots,N_{k1}$. Then the total number of remaining innovative packets $R$ is $R=\displaystyle\sum_{i=1}^{k}\sum_{j=1}^{\ell-1} R_{ij}$ and the expected time $\stexp\tilde{\tau}$ is upper bounded by
\begin{align}
\stexp\tilde{\tau}=\stexp\left[\stexp\left(\tilde{\tau}|R\right)\right]\leq \sum_{j=1}^{\ell}\frac{\stexp R}{1-p_{1j}}.
\label{eqn:stupid_remaining_packets}
\end{align}
where $\stexp R\slash (1-p_{1j})$ is the expected time taken for $R$ packets to cross the $j^{\text{th}}$ hop in the first path.

%
By combining the fact that $\stexp\tilde{T}_{n}^\text{c}\leq \stexp\tilde{\tau}$ with equations (\ref{eqn:proof_sum_of_expected_coding_times}) and (\ref{eqn:proof_expected_first_term})  we get
\begin{align}
\stexp T_{n}^\text{c}=\frac{n}{k-\displaystyle\sum_{i=1}^k p_{i1}}+D_{n}^\text{c}
\label{eqn:final_proof}
\end{align}
where $D^\text{c}_{n}$ is upper bounded by
\begin{align*}
D^\text{c}_{n}\leq B_{n}+\sum_{j=1}^{\ell}\frac{\stexp R}{1-p_{1j}}. 
\end{align*}
By Proposition~\ref{prop:The_markov_chain}, the number of remaining innovative packets at each node of each path is  a stochastically increasing random variable with respect to $n$. Therefore, the expected number of remaining packets is an increasing function of $n$. Consequently one can find an upper bound on $\stexp R_{ij}$ by examining the line network in steady state, or equivalently, as $n\rightarrow+\infty$. For the case where the first link of each path is the unique worst link of the path, as shown in~\cite{hsu76behavior}, each line network will reach steady state and consequently $E(R)\in\bigO(1)$. If there are multiple worst links in at least one path, then  $\stexp R\in\bigO(\sqrt{n})$. This can be seen by interchanging the positions of links such that the worst links of each path are positioned at the start. By the results of~\cite{pakzad05coding}, the number of innovative packets remaining at nodes positioned between two such worst links is $\bigO(\sqrt{n})$. By the results of~\cite{hsu76behavior}, the number of innovative packets remaining at other intermediate network nodes is $\bigO(1)$.

Substituting $p_{1i}$ with $\displaystyle\max_{1\leq j\leq\ell}p_{ij}$ for $i\in\{1,\ldots,k\}$ in equation (\ref{eqn:final_proof}) concludes the proof.\end{IEEEproof}

\subsection{Routing Strategy}

In this section  we analyze the expected time $\stexp T_{n}^\text{r}$ taken to send $n$ packets through the parallel path network in Figure~\ref{fig:parallel_multihop_network} using routing (where the r superscript stands for routing). We first prove the following two propositions.
\begin{prop}
For $a,b,c_1,c_2\in\mathbb{N^+}$ with $a<b$ the sum $\displaystyle\sum_{m=a}^b\frac{c_1-m}{c_2+m}$ is equal to:
\begin{eqnarray}
\sum_{m=a}^b\frac{c_1-m}{c_2+m} = a-b-1+(c_1+c_2)\left(H_{c_2+b}-H_{c_2+a-1}\right)
\label{eqn:prop:sum_of_nminusk_over_nplusk}
\end{eqnarray}
where $H_n$ is the $n^\text{th}$ Harmonic number, \textit{i.e.} $H_n=\displaystyle\sum_{i=1}^n\frac{1}{i}$.
\label{prop:sum_of_nminusk_over_nplusk}
\end{prop}
\begin{IEEEproof}
\begin{align}
\sum_{m=a}^b\frac{c_1-m}{c_2+m}=c_1\sum_{m=a}^b\frac{1}{c_2+m}- \sum_{m=a}^b\frac{m}{c_2+m}=c_1\left(H_{c_2+b}-H_{c_2+a-1}\right)-\sum_{m=a}^b\frac{m}{c_2+m}
\label{eqn:prop:first}
\end{align}
Where $\displaystyle\sum_{m=a}^b\frac{m}{c_2+m}$ can be evaluated as follows:
\begin{align}
b-a+1 &= \sum_{m=a}^{b}\frac{c_2+m}{c_2+m}\notag\\
\Leftrightarrow b-a+1 &= c_2\sum_{m=a}^{b}\frac{1}{c_2+m}+\sum_{m=a}^{b}\frac{m}{c_2+m}\notag\\
\Leftrightarrow \sum_{m=a}^{b}\frac{m}{c_2+m}&=b-a+1-c_2\left(H_{c_2+b}-H_{c_2+a-1}\right)
\label{eqn:prop:second}
\end{align}
So from equations (\ref{eqn:prop:first}) and (\ref{eqn:prop:second}) we conclude that:
\begin{align*}
\sum_{m=a}^b\frac{c_1-m}{c_2+m}=a-b-1+(c_1+c_2)\left(H_{c_2+b}-H_{c_2+a-1}\right)
\end{align*}
\end{IEEEproof}

Consider the network of Figure~\ref{fig:parallelqueues} with $k=2$ parallel erasure links. As shown in equation (\ref{eqn:minimize_routing_time}) in order to minimize the expected completed time the routing strategy sends $\frac{n(1-q_1)}{2-q_1-q_2}$ packets over the first link and $\frac{n(1-q_2)}{2-q_1-q_2}$ packets over the second link. Proposition~\ref{prop:two_parallel_links_routing} examines this expected transmission time under routing.
\begin{prop}
The expected time $\stexp \hat{T}_{n}^\text{r}$ taken to send by routing $n$ packets from the source to the destination through two parallel erasure links with probabilities of erasure $q_1$ and $q_2$ respectively is
\begin{align*}
\stexp \hat{T}_{n}^\text{r} = \frac{n}{2-q_1-q_2}+U_{n}^{q_1,q_2}
\end{align*}
where $U_n^{q_1,q_2}$ is an unbounded term that grows at least as square root of $n$. The term routing means that out of the $n$ packets, $\frac{n(1-q_1)}{2-q_1-q_2}$ packets are transmitted through the link with $q_1$ probability of erasure and $\frac{n(1-q_2)}{2-q_1-q_2}$ packets through the link with $q_2$ probability of erasure.
\label{prop:two_parallel_links_routing}
\end{prop}
\begin{IEEEproof}
Denote by $A_{i,j}$ the expected time to send $i$ packets over the link with erasure probability $q_1$ and $j$ packets over the link with erasure probability $q_2$. Clearly $\stexp\hat{T}_{n}^\text{r}=A_{i,j}$ with $i=\frac{n(1-q_1)}{2-q_1-q_2}$, $j=\frac{n(1-q_2)}{2-q_1-q_2}$. $A_{i,j}$ satisfies the following two dimensional recursion formula:
\begin{align*}
\left\{
\begin{array}{c}
A_{i,j}=q_1q_2[A_{i,j}+1]+(1-q_1)q_2[A_{i-1,j}+1]\\
+q_1(1-q_2)[A_{i,j-1}+1]+(1-q_1)(1-q_2)[A_{i-1,j-1}+1]\\
A_{i,0}=\frac{i}{1-q_1}, \quad A_{0,j}=\frac{j}{1-q_2}, \quad A_{0,0}=0
\end{array}
\right\}
\end{align*}
or equivalently
\begin{align}
\left\{\hspace{-1mm}
\begin{array}{c}
(1-q_1q_2)A_{i,j} = (1-q_1)q_2A_{i-1,j}+q_1(1-q_2)A_{i,j-1}\\
+(1-q_1)(1-q_2)A_{i-1,j-1}+1\\
A_{i,0}=\frac{i}{1-q_1}, \quad A_{0,j}=\frac{j}{1-q_2}, \quad A_{0,0}=0
\end{array}\hspace{-1mm}
\right\}.
\label{eqn:proof_recursion1}
\end{align}
The two dimensional recursion formula in (\ref{eqn:proof_recursion1}) has a specific solution $\frac{i}{2(1-q_1)}+\frac{j}{2(1-q_2)}$ and a general solution $B_{i,j}$ where
\begin{align}
\left\{\hspace{-1mm}
\begin{array}{c}
(1-q_1q_2)B_{i,j} = (1-q_1)q_2 B_{i-1,j}+q_1(1-q_2)B_{i,j-1}\\
+(1-q_1)(1-q_2)B_{i-1,j-1},i,j\geq 1\\
B_{i,0}=\frac{i}{2(1-q_1)}, \quad B_{0,j}=\frac{j}{2(1-q_2)}, \quad B_{0,0}=0
\end{array}\hspace{-1mm}
\right\}.
\label{eqn:proof_recursion2}
\end{align}

In order to solve equation (\ref{eqn:proof_recursion2}) we will use the $Z$--transform with respect to $i$. More specifically we define the $Z$--transform as:
\begin{align}
\hat{B}_{z,j}=\sum_{i=0}^{\infty}B_{i,j}\cdot z^i.
\label{eqn:proof_Z_transform}
\end{align}
By multiplying all terms in equation (\ref{eqn:proof_recursion2}) by $z^i$ and summing over $i$ we get:
\begin{align*}
(1-q_1q_2)\sum_{i=1}^{\infty} B_{i,j}\cdot z^i = (1-q_1)q_2\sum_{i=1}^{\infty} B_{i-1,j}\cdot z^i +q_1(1-q_2)\sum_{i=1}^{\infty} B_{i,j-1}\cdot z^i\notag\\
+(1-q_1)(1-q_2)\sum_{i=1}^{\infty} B_{i-1,j-1}\cdot z^i\qquad\qquad\qquad\quad\notag\\
\Leftrightarrow (1-q_1q_2)\left[\hat{B}_{z,j}-B_{0,j}\right]=z(1-q_1)q_2 \hat{B}_{z,j}+q_1(1-q_2)\left[\hat{B}_{z,j-1}-B_{0,j-1}\right]\notag\\
+z(1-q_1)(1-q_2)\hat{B}_{z,j-1}.\qquad\qquad\qquad\qquad\qquad\quad
\end{align*}
Since $B_{0,j}=\frac{j}{2(1-q_2)}$ the above equation becomes:
\begin{align}\hspace{-1mm}
\left\{
\begin{array}{c}
\left[(1-q_1q_2)-z(1-q_1)q_2\right]\hat{B}_{z,j}= \left[q_1(1-q_2)+z(1-q_1)(1-q_2)\right]\hat{B}_{z,j-1}\\
+j\frac{1-q_1}{2(1-q_2)}+\frac{q_1}{2}\\
\hat{B}_{z,0}=\sum_{i=0}^{\infty}B_{i,0} z^i=\sum_{i=0}^{\infty}\frac{i}{2(1-q_1)}z^i=\frac{z}{2(1-q_1)(1-z)^2}
\end{array}\hspace{-1mm}\right\}
\label{eqn:recursion_Z_transform1}
\end{align}
where equation (\ref{eqn:recursion_Z_transform1}) is an one dimensional recursion formula with the following general solution~\cite[Section $3.2$]{recursion:bk}:
\begin{align}
\hat{B}_{z,j}=\frac{z}{(1-q_1)(1-z)^2} \left[\frac{q_1(1-q_2)+z(1-q_1)(1-q_2)}{1-q_1q_2-z(1-q_1)q_2}\right]^j\notag\\ +\frac{j}{2(1-q_2)(1-z)}-\frac{z}{2(1-q_1)(1-z)^2}.\quad\qquad
\label{eqn:solution_of_recursion_formula}
\end{align}
\begin{table}
\centering
\caption{Some pairs of functions along with their $Z$--transforms}
\begin{tabular}{c c}
\hline\hline
     Sequence\quad&\quad Z--transform\\
         $1$            &      $\displaystyle\frac{1}{1-z}$\\
         $i$            &      $\displaystyle\frac{z}{(1-z)^2}$\\
$\frac{\nchoosek{i+j-t-1}{j-1}}{\displaystyle b^{i+j-t}}$  & \quad $\displaystyle\frac{z^t}{(b-z)^j}$, for $t\leq j$\\
\hline\hline
\end{tabular}
\label{table:Z_transform_pairs}
\end{table}

\noindent Equation (\ref{eqn:solution_of_recursion_formula}) can be written in a compact form
\begin{align}
\hat{B}_{z,j}=\hat{a}(z)\cdot \hat{b}(j,z)+\hat{d}(j,z)
\label{eqn:solution_of_recursion_formula_compact}
\end{align}
by defining the functions $\hat{a}(z)$, $\hat{b}(z,j)$ and $\hat{d}(z,j)$ as follows:
\begin{align*}
\hat{a}(z) &= \frac{z}{(1-q_1)(1-z)^2}\\
\hat{b}(z,j) &= \left[\frac{q_1(1-q_2)+z(1-q_1)(1-q_2)}{1-q_1q_2-z(1-q_1)q_2}\right]^j\\
\hat{d}(z,j) &= \frac{j}{2(1-q_2)(1-z)}-\frac{z}{2(1-q_1)(1-z)^2}.
\end{align*}

Now we are ready to compute the inverse $Z$--transform of $\hat{B}_{z,j}$. Using Table~\ref{table:Z_transform_pairs} along with equation~(\ref{eqn:solution_of_recursion_formula_compact}):
\begin{align*}
B_{i,j} &= Z^{-1}\left\{\hat{a}(z)\cdot \hat{b}(z,j)\right\}+ Z^{-1}\left\{\hat{d}(z,j)\right\}\notag\\
\Leftrightarrow B_{i,j} &= \sum_{m=0}^{i} a(i-m)\cdot b(m,j)+\frac{j}{2(1-q_2)}-\frac{i}{2(1-q_1)}
\end{align*}
where $a(i)$ and  $b(i,j)$ are the inverse $Z$--transforms of $\hat{a}(z)$ and $\hat{b}(z,j)$ respectively. From Table~\ref{table:Z_transform_pairs} $a(i)=\frac{i}{1-q_1}$ and therefore the equation above becomes
\begin{align}
B_{i,j} &= \sum_{m=0}^{i} \frac{i-m}{1-q_1}b(m,j) +\frac{j}{2(1-q_2)}-\frac{i}{2(1-q_1)}.
\label{eqn:equation_of_solution_of_B}
\end{align}
The remaining step in order to compute $B_{i,j}$ is to evaluate $b(i,j)$:
\begin{align*}
b(i,j) &= Z^{-1}\left\{\left[\frac{q_1(1-q_2)+z(1-q_1)(1-q_2)}{1-q_1q_2-z(1-q_1)q_2}\right]^j\right\}\\
&= \frac{1}{[(1-q_1)q_2]^j}\cdot Z^{-1}\left\{\frac{\sum_{t=0}^{j}\nchoosek{j}{t}z^t(1-q_1)^t(1-q_2)^t[q_1(1-q_2)]^{j-t}}{\left(\frac{1-q_1q_2}{(1-q_1)q_2}-z\right)^j}\right\}\\
&=\left[\frac{q_1(1-q_2)}{q_2(1-q_1)}\right]^j\sum_{t=0}^{j}\nchoosek{j}{t}\cdot \left(\frac{1-q_1}{q_1}\right)^t\cdot Z^{-1}\left\{\frac{z^t}{\left(\frac{1-q_1q_2}{(1-q_1)q_2}-z\right)^j}\right\}\\
&=\frac{(q_1(1-q_2))^j((1-q_1)q_2)^i}{(1-q_1q_2)^{i+j}}\sum_{t=0}^{j}\nchoosek{j}{t} \nchoosek{i+j-t-1}{j-1} \left(\frac{1-q_1q_2}{q_1q_2}\right)^t.
\end{align*}
Therefore equation (\ref{eqn:equation_of_solution_of_B}) becomes
\begin{align*}
B_{i,j} =\left[\frac{q_1(1-q_2)}{1-q_1q_2}\right]^j\sum_{m=0}^i\sum_{t=0}^j \frac{i-m}{1-q_1}
\left[\frac{(1-q_1)q_2}{1-q_1q_2}\right]^m\nchoosek{j}{t} \nchoosek{m+j-t-1}{j-1} \left(\frac{1-q_1q_2}{q_1q_2}\right)^t+\frac{j}{2(1-q_2)}-\frac{i}{2(1-q_1)}\hspace{55mm}
\end{align*}
and since the expected time $A_{i,j}=B_{i,j}+\frac{i}{2(1-q_1)}+\frac{j}{2(1-q_2)}$ then
\begin{align}
A_{i,j} =\left[\frac{q_1(1-q_2)}{1-q_1q_2}\right]^j\sum_{m=0}^i\sum_{t=0}^j \frac{i-m}{1-q_1}
\left[\frac{(1-q_1)q_2}{1-q_1q_2}\right]^m\nchoosek{j}{t} \nchoosek{m+j-t-1}{j-1} \left(\frac{1-q_1q_2}{q_1q_2}\right)^t+\frac{j}{1-q_2}.
\label{eqn:solution_of_A}
\end{align}

We are interested in evaluating $\stexp\hat{T}_{n}^\text{r}=A_{i,j}$ for $i=\frac{n(1-q_1)}{2-q_1-q_2}$ and $j=\frac{n(1-q_2)}{2-q_1-q_2}$ and therefore from equation (\ref{eqn:solution_of_A}) we get
\begin{align*}
\stexp \hat{T}_{n}^\text{r}=\frac{n}{2-q_1-q_2}+U_{n}^{q_1,q_2}
\end{align*}
where
\begin{align*}
U_n^{q_1,q_2}=\left[\frac{q_1(1-q_2)}{1-q_1q_2}\right]^{\frac{n(1-q_2)}{2-q_1-q_2}}\sum_{m=0}^{\frac{n(1-q_1)}{2-q_1-q_2}} \sum_{t=0}^{\frac{n(1-q_2)}{2-q_1-q_2}} \frac{\frac{n(1-q_1)}{2-q_1-q_2}-m}{1-q_1}
\left[\frac{(1-q_1)q_2}{1-q_1q_2}\right]^m\hspace{-1mm}\nchoosek{\frac{n(1-q_2)}{2-q_1-q_2}}{t} \hspace{-1mm}\nchoosek{m+\frac{n(1-q_2)}{2-q_1-q_2}-t-1}{\frac{n(1-q_2)}{2-q_1-q_2}-1} \hspace{-1mm}\left(\frac{1-q_1q_2}{q_1q_2}\right)^t
\end{align*}
with $\nchoosek{m}{w}=0$ if $m<w$. If we define $W=\frac{(1-q_1)q_2}{1-q_1q_2}$, $E=\frac{q_1(1-q_2)}{1-q_1q_2}$ and  $F=\frac{1-q_1q_2}{q_1q_2}$, then the above expression can be written more compactly as
\begin{align*}
U_n^{q_1,q_2}=E^{\frac{n(1-q_2)}{2-q_1-q_2}}\sum_{m=0}^{\frac{n(1-q_1)}{2-q_1-q_2}}\sum_{t=0}^{\frac{n(1-q_2)}{2-q_1-q_2}} \frac{\frac{n(1-q_1)}{2-q_1-q_2}-m}{1-q_1}\nchoosek{\frac{n(1-q_2)}{2-q_1-q_2}}{t} \nchoosek{\frac{n(1-q_2)}{2-q_1-q_2}+m-t-1}{\frac{n(1-q_2)}{2-q_1-q_2}-1}W^m F^t.
\end{align*}
In order to prove that function $U_n^{q_1,q_2}$ is unbounded we will prove that $U_n^{q_1,q_2}$ is larger than another simpler to analyze function that goes to infinity and therefore $U_n^{q_1,q_2}$ also increases to infinity. Indeed the equation above can be written as
\begin{align*}
U_n^{q_1,q_2}&=E^{\frac{n(1-q_2)}{2-q_1-q_2}}\sum_{m=0}^{\frac{n(1-q_1)}{2-q_1-q_2}}\sum_{t=0}^{\frac{n(1-q_2)}{2-q_1-q_2}} \frac{\frac{n(1-q_1)}{2-q_1-q_2}-m}{1-q_1}\nchoosek{\frac{n(1-q_2)}{2-q_1-q_2}}{t} \nchoosek{\frac{n(1-q_2)}{2-q_1-q_2}+m-t}{\frac{n(1-q_2)}{2-q_1-q_2}}\frac{\frac{n(1-q_2)}{2-q_1-q_2}}{\frac{n(1-q_2)}{2-q_1-q_2}+m-t} W^m F^t\\
&>\frac{n(1-q_2)E^{\frac{n(1-q_2)}{2-q_1-q_2}}}{(1-q_1)(2-q_1-q_2)}\sum_{m=0}^{\frac{n(1-q_1)}{2-q_1-q_2}}\sum_{t=0}^{\frac{n(1-q_2)}{2-q_1-q_2}}\nchoosek{\frac{n(1-q_2)}{2-q_1-q_2}}{t} \nchoosek{\frac{n(1-q_2)}{2-q_1-q_2}+m-t}{\frac{n(1-q_2)}{2-q_1-q_2}}\frac{\frac{n(1-q_1)}{2-q_1-q_2}-m}{\frac{n(1-q_2)}{2-q_1-q_2}+m} W^m F^t
\end{align*}
and since all terms in the above double sum are non-negative we can disregard as many terms as we wish without violating direction of the inequality, specifically
\begin{eqnarray}
U_n^{q_1,q_2}>\frac{n(1-q_2)E^{\frac{n(1-q_2)}{2-q_1-q_2}}}{(1-q_1)(2-q_1-q_2)}\sum_{m\in J, t\in G}\nchoosek{\frac{n(1-q_2)}{2-q_1-q_2}}{t} \nchoosek{\frac{n(1-q_2)}{2-q_1-q_2}+m-t}{\frac{n(1-q_2)}{2-q_1-q_2}}\frac{\frac{n(1-q_1)}{2-q_1-q_2}-m}{\frac{n(1-q_2)}{2-q_1-q_2}+m} W^m F^t
\label{eqn:U_n_chopped}
\end{eqnarray}
where $J=\{\lceil \frac{n(1-q_1)}{2-q_1-q_2}(1-\frac{1}{\sqrt{n}})\rceil,\ldots,\frac{n(1-q_1)}{2-q_1-q_2}\}$, $G=\{\lceil(1-q_1)\frac{n(1-q_2)}{2-q_1-q_2}(1-\frac{1}{\sqrt{n}})\rceil,\ldots,\lfloor(1-q_1)\frac{n(1-q_2)}{2-q_1-q_2}\rfloor\}$ and $\lfloor x \rfloor$, $\lceil x \rceil$ are the floor and the ceiling functions respectively.

By using the lower and upper Stirling-based bound \cite{Beesack}:
\begin{align*}
\sqrt{2\pi n}\left(\frac{n}{e}\right)^n < n! < \sqrt{2\pi n}\left(\frac{n}{e}\right)^n e^\frac{1}{12 n}, \quad n\geq 1
\end{align*}
one can find that
\begin{align*}
\nchoosek{n}{\beta n}\hspace{-1mm} >\hspace{-1mm} \frac{1}{\sqrt{2\pi\beta(1-\beta)n}} \cdot 2^{nH(\beta)}\cdot e^{-\frac{1}{12n\beta(1-\beta)}},\hspace{1mm} \beta\in(0,1)
\end{align*}
and
\begin{align*}
\nchoosek{\bar{\beta}n}{n} > \sqrt{\frac{\bar{\beta}}{2\pi(\bar{\beta}-1)n}} \cdot 2^{n \bar{\beta} H\left(\frac{1}{\bar{\beta}}\right)}\cdot e^{-\frac{\bar{\beta}}{12n(\bar{\beta}-1)}},\hspace{1mm} \bar{\beta}>1
\end{align*}
where $H(\beta)=-\beta\log_2(\beta)-(1-\beta)\log_2(1-\beta)$ is the entropy function and therefore using inequality (\ref{eqn:U_n_chopped}) we can derive:
\begin{align}
U_n^{q_1,q_2}>\frac{1}{2\pi(1-q_1)}\sum_{m\in J, t\in G}\frac{\frac{n(1-q_1)}{2-q_1-q_2}-m}{\frac{n(1-q_2)}{2-q_1-q_2}+m} f\left(\frac{m}{M},\frac{t}{T}\right)e^{-\frac{2-q_1-q_2}{12n(1-q_2)}h\left(\frac{m}{M},\frac{t}{T}\right)} 2^{\frac{n(1-q_2)}{2-q_1-q_2}g\left(\frac{m}{M},\frac{t}{T}\right)}
\label{eqn:U_n_chopped_abfirst}
\end{align}
where $M=\frac{n(1-q_1)}{2-q_1-q_2}$, $T=\frac{n(1-q_2)}{2-q_1-q_2}$, $f(\alpha,\beta)=\sqrt{\frac{1+\alpha\frac{1-q_1}{1-q_2}-\beta}{\beta(1-\beta)(\alpha\frac{1-q_1}{1-q_2}-\beta)}}$, $h(\alpha,\beta)=\frac{1+\alpha\frac{1-q_1}{1-q_2}-\beta}{\alpha\frac{1-q_1}{1-q_2}-\beta}+\frac{1}{\beta(1-\beta)}$ and
\begin{align*}
g(\alpha,\beta)=\log_2(E)+\alpha\frac{1-q_1}{1-q_2}\log_2\left(W\right)+H(\beta)+(1+\alpha\frac{1-q_1}{1-q_2}-\beta) H\left(\frac{1}{1+\alpha\frac{1-q_1}{1-q_2}-\beta}\right)+\beta\log_2(F).
\end{align*}
Since $1-\frac{1}{\sqrt{n}}\leq\frac{m}{M}\leq 1$ and $(1-q_1)-\frac{1}{\sqrt{n}}\leq\frac{t}{T}\leq(1-q_1)$ we define functions $f(\alpha,\beta)$, $h(\alpha,\beta)$ and $g(\alpha,\beta)$ within the region $N=\left[1-\frac{1}{\sqrt{n}},1\right]\times \left[1-q_1-\frac{1}{\sqrt{n}},1-q_1\right]$. Moreover we are only concerned with large enough $n$ so that $0<\beta<\alpha$ and region $N$ looks like the one in Figure~\ref{fig:region}. For large values of $n$, $f(\alpha,\beta)>\sqrt{\frac{1}{2q_1(1-q_1)}}$ and $h(\alpha,\beta)<1+\frac{2(1-q_2)}{(1-q_1)q_2}+\frac{2}{q_1(1-q_1)}$ within region $N$ and therefore from inequality (\ref{eqn:U_n_chopped_abfirst}) we get:
\begin{align}
U_n^{q_1,q_2}&>\frac{1}{\sqrt{8\pi^2q_1(1-q_1)^3}}e^{-\frac{2-q_1-q_2}{12n(1-q_2)}(1+\frac{2(1-q_2)}{(1-q_1)q_2}+\frac{2}{q_1(1-q_1)})}\sum_{m\in J, t\in G}\frac{\frac{n(1-q_1)}{2-q_1-q_2}-m}{\frac{n(1-q_2)}{2-q_1-q_2}+m} 2^{\frac{n(1-q_2)}{2-q_1-q_2}g\left(\frac{m}{M},\frac{t}{T}\right)}\notag\\
&>\frac{e^{-1}}{\sqrt{8\pi^2q_1(1-q_1)^3}}\sum_{m\in J, t\in G}\frac{\frac{n(1-q_1)}{2-q_1-q_2}-m}{\frac{n(1-q_2)}{2-q_1-q_2}+m} 2^{\frac{n(1-q_2)}{2-q_1-q_2}g\left(\frac{m}{M},\frac{t}{T}\right)}\hspace{20mm}
\label{eqn:U_n_chopped_absecond}
\end{align}
for large enough $n$.

\begin{figure}
\begin{center}
\includegraphics[width=0.6\columnwidth]{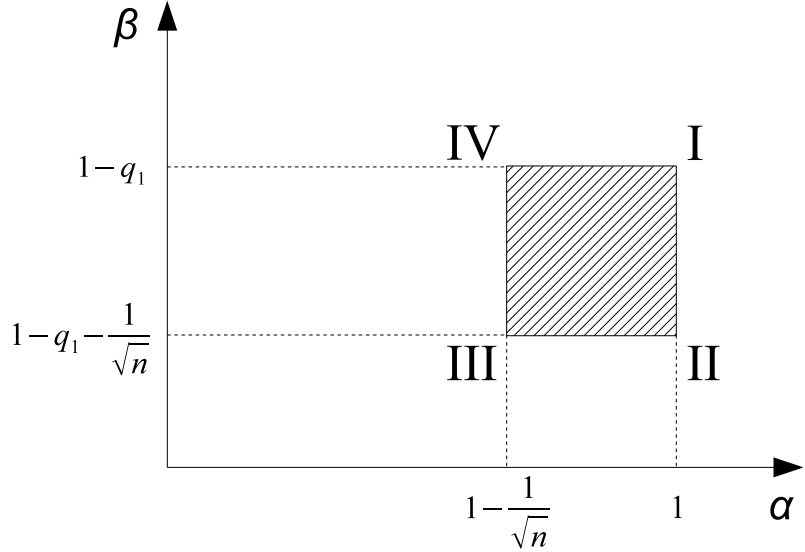}
\end{center}
\caption{The region $N$ where function $g(\alpha,\beta)$ is defined on.}
\label{fig:region}
\end{figure}

Function $g(\alpha,\beta)$ satisfies the following three conditions:
\begin{enumerate}
\item $\frac{\partial g}{\partial \alpha} =\frac{1-q_1}{1-q_2} \log_2\left(W\frac{\alpha(1-q_1)+(1-\beta)(1-q_2)}{\alpha(1-q_1)-\beta(1-q_2)}\right)$ and $\frac{\partial g}{\partial \beta} = \log_2\left(\frac{F(1-\beta)[\alpha(1-q_1)-\beta(1-q_2)]}{\beta[\alpha(1-q_1)+(1-\beta)(1-q_2)]}\right)$
\item $\frac{\partial^2 g}{\partial \alpha^2} = -\frac{(1-q_1)^2}{[\alpha(1-q_1)-\beta(1-q_2)][\alpha(1-q_1)+(1-\beta)(1-q_2)]\ln 2}<0$
\item $\frac{\partial^2 g}{\partial \alpha^2}\cdot \frac{\partial^2 g}{\partial \beta^2}-\frac{\partial^2 g}{\partial\alpha\partial\beta}\cdot \frac{\partial^2 g}{\partial\beta\partial\alpha}=\frac{(1-q_1)^2}{\beta(1-\beta)[\alpha(1-q_1)+(1-\beta)(1-q_2)] [\alpha(1-q_1)-\beta(1-q_2)](\ln 2)^2}>0$
\end{enumerate}

It's easy to see from condition $1$ that $\left.\frac{\partial g(\alpha,\beta)}{\partial \alpha}\right|_{(1,1-q_1)}=0$ and $\left.\frac{\partial g(\alpha,\beta)}{\partial \beta}\right|_{(1,1-q_1)}=0$. Moreover conditions $2$ and $3$ show the concavity of $g(\alpha,\beta)$ within region $N$ and along with condition $1$ it is proved that function $g(\alpha,\beta)$ achieves a maximum at point $(\alpha,\beta)=(1,1-q_1)$. Therefore $g(\alpha,\beta)\leq g(1,1-q_1)=0$ making the exponent of $2$ in (\ref{eqn:U_n_chopped_absecond}) non-positive guaranteeing an exponential decay of each term in the sum. Since region $N$ is compact (closed and convex) and function $g(\alpha,\beta)$ is concave, and therefore it will achieve its minimum on the boundary of $N$. It's not difficult to show that $\frac{\partial g(\alpha,1-q_1)}{\partial \alpha}\geq 0$ for $\alpha\leq 1$ and therefore function $g(\alpha,1-q_1)$ decreases in value from point I to point IV. Similarly $\frac{\partial g(1,\beta)}{\partial \beta}\geq 0$ for $\beta\leq 1-q_1$ and therefore function $g(1,\beta)$ decreases in value from point I to point II. Since $\frac{\partial g(\alpha,1-q_1-1/\sqrt{n})}{\partial \alpha}\geq 0$ for $a\leq 1$ and $\frac{\partial g(1-1/\sqrt{n},\beta)}{\partial \beta}\geq 0$ for $\beta\leq 1-q_1$ with similar arguments as above we show that the minimum value for $g(\alpha,\beta)$ within $N$ is achieved at point $C\equiv(\alpha_m,\beta_m)=(1-\frac{1}{\sqrt{n}},1-q_1-\frac{1}{\sqrt{n}})$. Therefore $g\left(\frac{k}{n},\frac{i}{n}\right)\geq g\left(\alpha_m, \beta_m\right)$ or else from equation~(\ref{eqn:U_n_chopped_absecond}):
\begin{align*}
U_n^{q_1,q_2}>\frac{e^{-1}(1-q_2)\sqrt{n}}{(2-q_1-q_2)\sqrt{8\pi^2q_1(1-q_1)}} 2^{\frac{n(1-q_2)}{2-q_1-q_2}g\left(a_m,\beta_m\right)}\sum_{m\in J} \frac{\frac{n(1-q_1)}{2-q_1-q_2}-m}{\frac{n(1-q_2)}{2-q_1-q_2}+m}
\end{align*}

Using the Taylor expansion of function $r(x)=g(1-x,1-q_1-x)$ around $x = 0$ we get the following expression:
\begin{align*}
f(x)=\frac{q_1^2(q_2-q_1)-q_2(1-q_1^2)}{(1-q_1)q_1q_2(1-q_1q_2)\ln 2}x^2+\bigO(x^3).
\end{align*}
For $x=\frac{1}{\sqrt{n}}$ we get
\begin{align*}
\frac{n(1-q_2)}{2-q_1-q_2}g\left(\alpha_m,\beta_m\right)=\frac{(1-q_2)\left(q_1^2(q_2-q_1)-q_2(1-q_1^2)\right)}{(2-q_1-q_2)(1-q_1)q_1q_2(1-q_1q_2)\ln 2}+\bigO\left(\frac{1}{\sqrt{n}}\right)
\end{align*}
where along with Proposition~\ref{prop:sum_of_nminusk_over_nplusk} we get
\begin{align}
U_n^{q_1,q_2}>\frac{e^{-1}(1-q_2)\sqrt{n}}{(2-q_1-q_2)\sqrt{8\pi^2q_1(1-q_1)}} 2^{\frac{(1-q_2)\left(q_1^2(q_2-q_1)-q_2(1-q_1^2)\right)}{(2-q_1-q_2)(1-q_1)q_1q_2(1-q_1q_2)\ln 2}+\frac{c}{\sqrt{n}}}t(n)
\label{eqn:U_n_chopped_abfourth}
\end{align}
where $t(n)=n\left(H_{n}-H_{n-k(n)-1}\right)- k(n)-1$ and $k(n)=A\sqrt{n}$ with $A = \frac{(1-q_1)}{2-q_1-q_2}$. The above expression can be simplified by using the bounds proved by Young in~\cite{Young_harmonic}:
\begin{eqnarray*}
\ln n+\gamma+\frac{1}{2(n+1)}<H_n<\ln n+\gamma+\frac{1}{2n}
\end{eqnarray*}
where $\gamma$ is the Euler's constant. We obtain from (\ref{eqn:U_n_chopped_abfourth}):
\begin{align}
U_n^{q_1,q_2}>\frac{e^{-1}(1-q_2)\sqrt{n}}{(2-q_1-q_2)\sqrt{8\pi^2q_1(1-q_1)}} 2^{\frac{(1-q_2)\left(q_1^2(q_2-q_1)-q_2(1-q_1^2)\right)}{(2-q_1-q_2)(1-q_1)q_1q_2(1-q_1q_2)\ln 2}+\frac{c}{\sqrt{n}}}\phi(n)
\label{eqn:last}
\end{align}
where $\phi(n)=n\ln\left(\frac{n}{n-k(n)-1}\right)-\frac{n}{2(n+1)}\frac{k(n)+2}{n-k(n)-1}-k(n)-1$. It can be easily proved that function $\omega(n)=n\ln\left(\frac{n}{n-k(n)-1}\right)-k(n)-1$ is greater than $\frac{A^2}{2}$ for $n>1$. Indeed
\begin{eqnarray}
\omega''(n)=\frac{A(A^2+3)n+2(A^2+2)\sqrt{n}+A}{4(n-A\sqrt{n}-1)^2n^{3/2}}>0\quad\text{for}\quad n>1
\end{eqnarray}
and since $\displaystyle\lim_{n\rightarrow +\infty} \omega'(n)=0$ it means that $\omega'(n)<0$ for $n>1$ and therefore $\omega(n)$ is a decreasing function of $n>1$. Moreover
\begin{eqnarray*}
\lim_{n\rightarrow +\infty}\omega(n)=\lim_{n\rightarrow +\infty}\frac{\ln\left(\frac{n}{n-k(n)-1}\right)-\frac{k(n)}{n}}{\frac{1}{n}}-1 \mathop{=}^{\text{L'Hospital}}\lim_{n\rightarrow +\infty}\frac{\frac{k(n)}{n^2}+\frac{k^2(n)}{n^2}+\frac{2}{n}}{-\frac{1}{n^2}(2+2k(n)-2n)}-1=\frac{A^2}{2}
\end{eqnarray*}
and therefore $\omega(n)>\frac{A^2}{2}$ for $n>1$. Finally inequality (\ref{eqn:last}) becomes
\begin{align*}
U_n^{q_1,q_2}>\frac{e^{-1}(1-q_2)\sqrt{n}}{(2-q_1-q_2)\sqrt{8\pi^2q_1(1-q_1)}} 2^{\frac{(1-q_2)\left(q_1^2(q_2-q_1)-q_2(1-q_1^2)\right)}{(2-q_1-q_2)(1-q_1)q_1q_2(1-q_1q_2)\ln 2}+\frac{c}{\sqrt{n}}}\left(\frac{1}{2}\left(\frac{1-q_1}{2-q_1-q_2}\right)^2-\frac{n}{2(n+1)}\frac{k(n)+2}{n-k(n)-1}\right).
\end{align*}
Clearly the above function is unbounded and $U_n^{q_1,q_2}$ increases with respect to $n$ at least as $\sqrt{n}$.
\end{IEEEproof}
Now we have all the necessary tools to prove the following theorem for  $k$-parallel path multi-hop networks as shown in Figure~\ref{fig:parallel_multihop_network}.
\begin{IEEEproof}[Proof of Theorem~\ref{thm:routing}]
Without loss of generality due to~\cite{weber92interchangeability} we can interchange the first link of each of the $k$ line networks with the worst link of the line network. The first term in equation (\ref{eqn:finish}) is due to the capacity of the $k$ parallel multi-hop line network. The second term $D^\text{r}_n$ is sublinear in $n$; 
what is left to prove is that term $D^\text{r}_n$ grows as $\Omega(\sqrt{n})$. This follows from Proposition~\ref{prop:two_parallel_links_routing}. The number of packets transmitted on the first two paths is $n_ 1 =n\Big(1-\displaystyle\mathop{\max}_{1\leq i\leq\ell}p_{1i}\Big)\Big\slash\Big(k-\displaystyle\sum_{i=1}^k\mathop{\max}_{1\leq j\leq\ell}p_{ij}\Big)$ and $n_2 =n\Big(1-\displaystyle\mathop{\max}_{1\leq i\leq\ell}p_{2i}\Big)\Big\slash\Big(k-\displaystyle\sum_{i=1}^k\mathop{\max}_{1\leq j\leq\ell}p_{ij}\Big)$ respectively. The time $T_n^\text{r}$ taken to send $n$ packets through  the $k$-parallel path multi-hop  network 
is greater than the time $\hat{T}_{n}^\text{r}$ taken for $n_1$ packets to reach node $N_{11}$ and $n_2$ packets to reach node $N_{21}$. Therefore 
from Proposition~\ref{prop:two_parallel_links_routing}
\begin{align*}
\stexp T_n^\text{r}>\frac{n}{k-\displaystyle\sum_{i=1}^k\mathop{\max}_{1\leq j\leq\ell}p_{ij}}+U_{n'}^{\mathop{\max}_{1\leq i\leq\ell}p_{1i}, \mathop{\max}_{1\leq j\leq\ell}p_{2j}}.
\end{align*}
where $n'=n\Big(2-\displaystyle\mathop{\max}_{1\leq i\leq\ell}p_{1i}-\displaystyle\mathop{\max}_{1\leq i\leq\ell}p_{2i}\Big)\Big\slash\Big(k-\displaystyle\sum_{i=1}^k\mathop{\max}_{1\leq j\leq\ell}p_{ij}\Big)$ is proportional to $n$. By Proposition~\ref{prop:two_parallel_links_routing}, $U_{n'}^{\mathop{\max}_{1\leq i\leq\ell}p_{1i}, \mathop{\max}_{1\leq j\leq\ell}p_{2j}}$ grows as $\Omega(\sqrt{n'})$.  Thus, $D^\text{r}_n$ grows as $\Omega(\sqrt{n})$.
\end{IEEEproof}
\section{General network topologies}
We next consider networks with general topologies.
\begin{lemma} In a single-bottleneck network, there exists a max-flow subgraph comprising paths each of which has a single worst link.
\label{lemma:bottleneck}
\end{lemma}\begin{IEEEproof}
Given a network $\mathcal{G}=(\mathcal{V},\mathcal{E})$ with a single minimum cut, let $(v_1,w_1),\dots,(v_k,w_k)$ be the edges crossing the minimum cut. Let $\mathcal{G}'$ be a max flow subgraph. Consider the network $\mathcal{G} -\mathcal{G}'$  obtained from $\mathcal{G}$ by reducing the capacity of each link  $(i,j)\in\mathcal{E}$ by the capacity of the corresponding link in $\mathcal{G}'$ if any. There is a path from the source to  each node $v_i,1 \le i\le k$ (which may not all be distinct), otherwise this would contradict the assumption that there is a single minimum cut. Thus, we can find a subgraph $\mathcal{G}''$ comprising a set of paths of nonzero and nonoverlapping capacity from the source to each distinct node $v_i,1 \le i\le k$. Similarly, we can find a subgraph $\mathcal{G}'''$ comprising a set of paths of nonzero and nonoverlapping capacity from each distinct node $w_i,1 \le i\le k$, to the sink. We can then decompose the union of subgraphs $\mathcal{G}'+\mathcal{G}''+\mathcal{G}'''$ (obtained by adding the capacities of corresponding links) into a sufficiently large number of paths each of which has a single worst link corresponding to the min cut of the original network.
\end{IEEEproof}
%

%
%
%
%

\begin{IEEEproof}[Proof of Theorem~\ref{thm:super_general}]
The expected time $\stexp T_{n}^\text{r}$ required to send all $n$ packets by routing through network $\mathcal{G}$ from source $\mathcal{S}$ to destination $\mathcal{T}$ is greater than the time $\stexp \breve{T}_{n}^\text{r}$ it would take the $n$ packets to cross the mincut of the network by routing. Specifically if we assume that all nodes on the source's side of the cut are collapsed into a super source node and all nodes on the sink's side of the cut are collapsed into a super destination node then the network becomes a parallel erasure links network as shown in Figure~\ref{fig:parallelqueues}. Then
\begin{align*}
\stexp T_{n}^\text{r}\geq\stexp \breve{T}_{n}^\text{r}=\frac{n}{C}+D^{\text{r}}_n
\end{align*}
where $D^{\text{r}}_n\in\Omega(\sqrt{n})$ by Theorem~\ref{thm:routing}.
\begin{figure}
\begin{center}
\includegraphics[width=0.8\columnwidth]{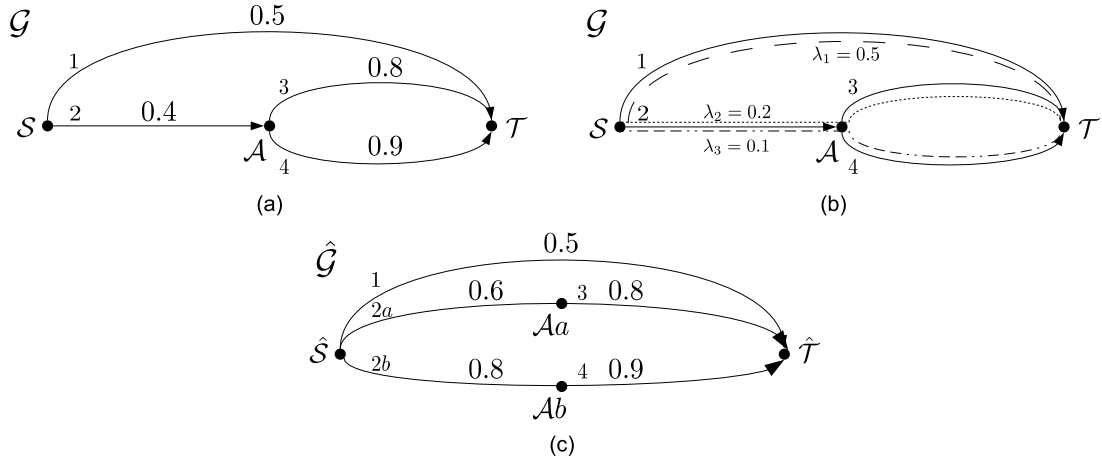}
\end{center}
\caption{(a) Network $\mathcal{G}$ with a single source $\mathcal{S}$, a single destination $\mathcal{T}$, an intermediate node $\mathcal{A}$, and four erasure links $1$, $2$, $3$, and $4$ with probabilities of erasure $0.5$, $0.4$, $0.8$, $0.9$ respectively. (b) The solution of the linear program on network $\mathcal{G}$ would give us three rates $\lambda_1=0.5$, $\lambda_2=0.2$, and $\lambda_3=0.1$. $(c)$ Network $\hat{\mathcal{G}}$ derived from the solution of the linear program}
\label{fig:super_general_network}
\end{figure}

For the case of coding on a network $\mathcal{G}$, for any max-flow subgraph (composed of flows on paths from source $\mathcal{S}$ to destination $\mathcal{T}$), 
one can construct a parallel path network $\hat{\mathcal{G}}$ that requires at least as much time to send the $n$ packets from the source to the destination.

Denote by $\mathcal{F}$ the set of source-sink flows in the max-flow subgraph. For each flow $f\in\mathcal{F}$, let $\lambda_f$ denote the flow rate and let $\mathcal{P}_f$ denote the path of flow $f$. For each node $v\in\mathcal{V}$ in network $\mathcal{G}$, let $\mathcal{K}_v$ denote the set of flows passing through node $v$, where $\mathcal{K}_{\mathcal{S}}$ and $\mathcal{K}_{\mathcal{T}}$ are equal to the sets of all flows in network $\mathcal{G}$. For each edge $e\in\mathcal{E}$ let $\mathcal{F}_e$ denote the set of flows passing through edge $e$. For the example in Figure~\ref{fig:super_general_network}(b),  $\mathcal{F}=\{1,2,3\}$, $\lambda_1=0.5$, $\lambda_2=0.2$, $\lambda_3=0.1$,  $\mathcal{P}_1=\mathcal{S}\rightarrow\mathcal{T}$ for flow $1$, $\mathcal{P}_2=\mathcal{S}\rightarrow\mathcal{A}\rightarrow\mathcal{T}$ for flow $2$, and $\mathcal{P}_3=\mathcal{S}\rightarrow\mathcal{A}\rightarrow\mathcal{T}$ for flow $3$,  $\mathcal{K}_{\mathcal{A}}=\{2,3\}$, and $\mathcal{F}_1=\{1\}$, $\mathcal{F}_2=\{2,3\}$,  $\mathcal{F}_3=\{2\}$, and  $\mathcal{F}_4=\{3\}$.

The process of creating network $\hat{\mathcal{G}}=(\hat{\mathcal{V}},\hat{\mathcal{E}})$ from $\mathcal{G}$ is the following.
\begin{enumerate}
\item For every node $v\in\mathcal{G}$, create a set of nodes $\hat{\mathcal{V}}_v=\big\{\hat{v}_f:\;f\in\mathcal{K}_v\big\}$. The set of nodes $\hat{\mathcal{V}} $ is defined as $\displaystyle\mathop{\bigcup}_{v\in\mathcal{V}}\hat{\mathcal{V}}_v$.

\item The edges of network $\hat{\mathcal{G}}$ are created as follows. For each flow $f\in\mathcal{F}$ and for each edge $(u, v)$ in path $\mathcal{P}_f$ of flow $f$, create an edge in network $\hat{\mathcal{G}}$ from $\hat{u}_f$ to $\hat{v}_f$ with probability of erasure
\begin{align*}
\hat{p}_{(\hat{u}_f,\hat{v}_f)}=1-\frac{\lambda_f}{\displaystyle\mathop{\sum}_{w\in\mathcal{F}_{(u,v)}}\lambda_{w}}(1-p_{(u,v)})
\end{align*}
where $p_{(u,v)}$ is the probability of erasure of link $(u,v)$ in network $\mathcal{G}$. Define a function $H\big((\hat{u}_f,\hat{v}_f)\big)=\frac{\lambda_f}{\sum_{w\in\mathcal{F}_{u,v)}}\lambda_{w}}$.

\item Collapse all nodes of set $\mathcal{V}_{\mathcal{S}}$ to a single node $\hat{\mathcal{S}}$ that denotes the source in network $\hat{\mathcal{G}}$, and collapse all nodes of set $\mathcal{V}_{\mathcal{T}}$ to a single node $\hat{\mathcal{T}}$ that denotes the destination in network $\hat{\mathcal{G}}$.
\end{enumerate}

The process above splits every node $v\in\mathcal{V}$ into $K_v$ separate nodes and splits every edge $e\in\mathcal{E}$ into $|\mathcal{F}_e|$ separate edges. The sum of capacities of all edges that edge $e$ is split into is equal to the capacity of edge $e$. The result of applying this procedure to network $\mathcal{N}$ of Figure~\ref{fig:super_general_network}(b) is shown in Figure~\ref{fig:super_general_network}(c). In network $\hat{\mathcal{G}}$ erasure events on different links are not independent but correlated as follows. For every edge $(u,v)\in\mathcal{E}$, denote by $\mathcal{C}_{(u,v)}=\big\{(\hat{u},\hat{v})\in\hat{\mathcal{E}}: \hat{v}\in\mathcal{K}_u, \hat{m}\in\mathcal{K}_v\big\}$ the set of edges in $\hat{\mathcal{G}}$ that are derived from edge $(u,v)\in\mathcal{E}$. 
The erasures on all edges in set $\mathcal{C}_{(u,v)}$ are not independent but correlated as follows. At each time step, with probability $1-p_{(u,v)}$ one edge in set $\mathcal{C}_{(u,v)}$ succeeds, or all fail with probability $p_{(u,v)}$. In the case of a success, edge $\hat{e}\in\mathcal{C}_{(u,v)}$ is the single successful edge with probability $H_{\hat{e}}$. 

The time taken $\hat{T}^{\text{c}}_n$ for the $n$ packets to travel through network $\hat{\mathcal{G}}$ by coding is at least as large as the time $T^{\text{c}}_n$ taken in network $\mathcal{G}$,  i.e.
\begin{align}
\stexp T^{\text{c}}_n\leq \stexp \hat{T}^{\text{c}}_n.
\label{eqn:almost_last_inequality}
\end{align}
Indeed network $\hat{\mathcal{G}}$ can be emulated by network $\mathcal{G}$ if each node $v\in\mathcal{G}$ has $|K_v|$ different buffers and packets between different buffers are not mixed. By construction, networks $\mathcal{G}$ and $\hat{\mathcal{G}}$ have the same capacity and since $\hat{\mathcal{G}}$ is a parallel path network, the mincut of network $\hat{\mathcal{G}}$ passes through the worst link of each path. According to Theorem~\ref{thm:coding}
\begin{align}
\stexp\hat{T}^{\text{c}}_n = \frac{n}{C}+\hat{D}^{\text{c}}_n
\label{eqn:last_inequality}
\end{align}
where $\hat{D}^{\text{c}}_n\in\Omega(\sqrt{n})$ when there are multiple worst links in at least one path or $\hat{D}^{\text{c}}_n\in\bigO(1)$ when there is a single worst link at each path. For a single-bottleneck network, by Lemma~\ref{lemma:bottleneck}, one can construct a max-flow subgraph comprising paths each of which has a single worst link, so $\hat{D}^{\text{c}}_n\in\bigO(1)$. Equations (\ref{eqn:almost_last_inequality}), (\ref{eqn:last_inequality}) conclude our proof.
\end{IEEEproof}

\section{Proof of concentration}
\label{sec:concentration}

Here we present a martingale concentration argument. In particular we prove a slightly stronger version of Theorem~\ref{thm:concentration}:
\begin{theorem}[Extended version of Theorem~\ref{thm:concentration}]
The time $T_n^\text{c}$ for $n$ packets to be transmitted from a source to a sink over a  network of erasure channels using network coding is concentrated around its expected value with high probability. In particular for sufficiently large $n$:
\begin{align*}
\Prob[|T_n^\text{c}-\stexp T_n^\text{c}|> \epsilon_{n}] \leq \frac{2C}{n}+\frac{2Cn^{2\delta}}{n^2-n^{1+2\delta}}.
\end{align*}
where $C$ is the capacity of the network and $\epsilon_n$ represents the corresponding deviation and is equal to $\epsilon_{n} = n^{1/2+\delta}\slash C$, $\delta\in (0,1/2)$.
\label{last_theorem}
\end{theorem}
\begin{IEEEproof}
The main idea of the proof is to use the method of Martingale bounded differences \cite{mitzenmacher05probability}. This method works as follows: first we show that the random variable we want to show is concentrated is a function of a finite set of independent random variables. Then we show that this function is Lipschitz with respect to these random variables, \emph{i.e.} it cannot change its value too much if only one of these variables is modified. Using this function we construct the corresponding Doob martingale and use the Azuma-Hoeffding~\cite{mitzenmacher05probability} inequality to establish concentration. See also \cite{azuma1, azuma2} for related concentration results using similar martingale techniques. Unfortunately however this method does not seem to be directly applicable to $T^\text{c}_n$ because it cannot be naturally expressed as a function of a \emph{bounded number} of independent random variables. We use the following trick of showing concentration for another quantity first and then linking that concentration to the concentration of $T^\text{c}_n$.

Specifically, we define $R_t$ to be the number of innovative (linearly independent) packets received at the destination node $\mathcal{T}$ after $t$ time steps. $R_t$ is linked with $T^\text{c}_n$ through the equation:
\begin{equation}
T_n^\text{c} =\displaystyle\mathop{\arg}_t(R_t = n)
\label{Relate_T_R}.
\end{equation}
The number of received packets is a well defined function of the link states at each time step. If there are $L$ number of links in network $\mathcal{G}$, then:
\begin{equation*}
R_t = g(z_{11},...,z_{1L},\ldots,z_{t1},...,z_{tL}).
\label{eqn:function_of_R}
\end{equation*}
The random variables $z_{ij}$,$1\leq i \leq t$ and $1\leq j \leq L$, are equal to $0$ or $1$ depending on whether link $j$ is OFF or ON at time $i$. If a packet is sent on a link that is ON, it is received successfully; if sent on a link that is OFF, it is erased. It is clear that this function satisfies a bounded Lipschitz condition with a bound equal to $1$:
\begin{align*}
|g(z_{11},...,z_{1L},..., z_{ij},...,z_{t1},...,z_{tL}) - \notag\\
g(z_{11},...,z_{1L},...,z_{ij}^{'} ,...,z_{t1},...,z_{tL})| \leq 1.
\label{eqn:g_inequality}
\end{align*}
This is because if we look at the history of all the links failing or succeeding at all the $t$ time slots, changing one of these link states in one time slot can at most influence the received rank by one.
We note that we assume that coding is performed over a very large field to ensure that every packet that could potentially be innovative due to connectivity, indeed is.

Using the Azuma-Hoeffding inequality (see the Appendix Theorem~\ref{thm:Azuma_Hoeffding_inequality}) on the Doob martingale constructed by $R_t = g(z_{11},...,z_{1L},...,z_{t1},...,z_{tL})$ we get following the concentration result:
\begin{prop}
The number of received innovative packets $R_t$ is a random variable concentrated around its mean value:
\begin{equation}
\Prob(|R_t-\stexp R_t|\geq \varepsilon_t)\leq \frac{1}{t}\hspace{2.5mm}\text{where}\hspace{2.5mm} \varepsilon_t\doteq \sqrt{\frac{t L}{2} \ell n(2t)}.
\label{eqn:concentration_of_R_2}
\end{equation}
\label{prop:concentration_of_Rt}
\end{prop}
\begin{IEEEproof}
Given in Appendix~\ref{appen:concentration}.
\end{IEEEproof}

Using this concentration and the relation (\ref{Relate_T_R}) between $T_n^\text{c}$ and $R_t$ we can show that deviations of the order $\varepsilon_t \doteq \sqrt{\frac{t L}{2} \ell n(2t)}$ for $R_t$ translate to deviations of the order of $\epsilon_{n} = n^{1/2+\delta}/C$ for $T_n^\text{c}$. In Theorem~\ref{last_theorem} smaller values $\delta$ give tighter bounds that hold for larger $n$. Define the events:
\begin{equation*}
H_t=\{|R_t-\stexp R_t|< \varepsilon_t\}
\label{eqn:H_event}
\end{equation*}
and
\begin{equation*}
\overline{H}_t=\{|R_{t}-\stexp R_t| \geq \varepsilon_t\}
\label{eqn:H_hat_event_1}
\end{equation*}
and further define $t^u_n$ ($u$ stands for upper bound) to be some $t$, ideally the smallest $t$, such that $\stexp R_t -\varepsilon_t\geq n$ and $t^l_n$ ($l$ stands for lower bound) to be some $t$, ideally the largest $t$, such that $\stexp R_t+\varepsilon_t\leq n$. Then we have:
\begin{eqnarray*}
\Prob(T_n^\text{c}\geq t^u_n) &=& \Prob(T_n^\text{c}\geq t^u_n|H_{t^u_n}) \cdot \Prob(H_{t^u_n})\notag\\
                           &+& \Prob(T_n^\text{c}\geq t^u_n|\overline{H}_{t^u_n})\cdot \Prob(\overline{H}_{t^u_n})
\end{eqnarray*}
where:
\begin{itemize}
  \item $\Prob(T_n^\text{c}\geq t^u_n|H_{t^u_n}) = 0$ since at time $t = t^u_n$ the destination has already received more than $n$ innovative packets. Indeed given that $H_{t^u_n}$ holds: $n\leq \stexp R_{t^u_n} - \varepsilon_{t^u_n} < R_{t^u_n}$ where the first inequality is due to the definition of $t^u_n$.
  \item $\Prob(H_{t^u_n})\leq 1$
  \item $\Prob(T_n^\text{c}\geq t^u_n|\overline{H}_{t^u_n}) \leq 1$
  \item $\Prob(\overline{H}_{t^u_n})\leq \frac{1}{t^u_n}$ due to equation (\ref{eqn:concentration_of_R_2}).
\end{itemize}
Therefore:
\begin{equation}
\Prob(T_n^\text{c}\geq t^{u}_{n})\leq \frac{1}{t^{u}_{n}}.
\label{ineq:t_u_n}
\end{equation}

Similarly:
\begin{eqnarray*}
\Prob(T_n^\text{c}\geq t^l_n) &=& \Prob(T_n^\text{c}\geq t^l_n|H_{t^l_n}) \cdot \Prob(H_{t^l_n})\notag\\
                           &+& \Prob(T_n^\text{c}\geq t^l_n|\overline{H}_{t^l_n})\cdot \Prob(\overline{H}_{t^l_n})
\end{eqnarray*}
where:
\begin{itemize}
  \item $\Prob(T_n^\text{c}\leq t^l_n|H_{t^l_n}) = 0$  since at time $t = t^l_n$ the destination has already received less than $n$ innovative packets. Indeed given that $H_{t^l_n}$ holds: $R_{t^u_n} < \stexp R_{t^u_n} + \varepsilon_{t^u_n} < n$ where the last inequality is due to the definition of $t^l_n$.
  \item $\Prob(H_{t^l_n})\leq 1$
  \item $\Prob(T_n^\text{c}\leq t^l_n|\overline{H}_{t^l_n}) \leq 1$
  \item $\Prob(\overline{H}_{t^l_n})\leq \frac{1}{t^l_n}$ due to equation (\ref{eqn:concentration_of_R_2}).
\end{itemize}
Therefore:
\begin{equation}
\Prob(T_n^\text{c}\leq t^l_n)\leq \frac{1}{t^l_n}.
\label{ineq:t_l_n}
\end{equation}

Equations (\ref{ineq:t_u_n}) and (\ref{ineq:t_l_n}) show that the random variable $T_n^\text{c}$ representing the time required for $n$ packets to travel across network $\mathcal{G}$ exhibits some kind of concentration between $t^l_n$ and $t^u_n$, which are both functions of $n$.
As shown in Lemma~\ref{lemma:expressing_t_up_and_t_l} in Appendix~\ref{appen:concentration}, for large enough $n$ a legitimate choice for $t^l_n$ and $t^u_n$ is the following:
\begin{equation}
t^u_n=(n+n^{1/2+\delta'})/C,\text{ } \delta' \in (0,1/2)
\label{t_u_n_delta}
\end{equation}
\begin{equation}
t^l_n=(n-n^{1/2+\delta'})/C,\text{ } \delta' \in (0,1/2)
\label{t_l_n_delta}
\end{equation}

From both (\ref{ineq:t_u_n}) and (\ref{ineq:t_l_n}):
\begin{eqnarray}
\Prob(t^l_n\leq T_n^\text{c}\leq t^u_n)&=&1-\Prob(T_n^\text{c} \leq t^l_n)-\Prob(T_n^\text{c}\geq t^u_n)\notag\\
&\geq& 1-\frac{1}{t^l_n}-\frac{1}{t^u_n}
\label{ineq:Prob_t_n_1}
\end{eqnarray}
and by substituting in (\ref{ineq:Prob_t_n_1}) the $t^u_n$, $t^l_n$ from equations (\ref{t_u_n_delta}) and (\ref{t_l_n_delta}) we get:
\begin{eqnarray*}
\Prob(-\frac{n^{1/2+\delta'}}{C}\leq T_n^\text{c}-\frac{n}{C}\leq \frac{n^{1/2+\delta'}}{A})\geq 1 -\notag\\
\frac{C}{n-n^{1/2+\delta'}}-\frac{C}{n+n^{1/2+\delta'}}
\end{eqnarray*}
and since $\stexp T_n^\text{c} = \frac{n}{C}+\bigO(\sqrt{n})$ we have:
\begin{equation*}
\Prob(|T_n^\text{c}-\stexp T_n^\text{c}|\leq \frac{n^{1/2+\delta}}{C})\geq 1-\frac{2C}{n}-\frac{2Cn^{2\delta}}{n^2-n^{1+2\delta}}
\end{equation*}
or
\begin{equation*}
\Prob(|T_n^\text{c}-\stexp T_n^\text{c}| > \frac{n^{1/2+\delta}}{C})\leq \frac{2C}{n}+\frac{2Cn^{2\delta}}{n^2-n^{1+2\delta}}
\end{equation*}
where $\delta > \delta'$ and this concludes the proof.
\end{IEEEproof}

\begin{appendices}
\section{Proof of Proposition~\ref{prop:The_markov_chain}}
\label{appen:Proof_of_stochastically_increasing}

\begin{dfn}
A binary relation $\preceq$ defined on a set $P$ is called a preorder if it is reflexive and transitive, i.e. $\forall a, b, c \in P$:
\begin{eqnarray}
&a\preceq a &\text{ (reflexivity)}\\
&(a\preceq b)\wedge(b\preceq c) \Rightarrow a\preceq c &\text{ (transitivity)}
\end{eqnarray}
\end{dfn}

\begin{dfn}
On the set $\mathbb{N}\hspace{0.8mm}^{\ell-1}$ of all integer $(\ell-1)$-tuples we define the regular preorder $\preceq$ that is $\forall a,b\in\mathbb{N}\hspace{0.8mm}^{\ell-1}$ $a\preceq b$ iff $a_1\leq b_1,\ldots,a_{\ell-1}\leq b_{\ell-1}$ where $a=(a_1,\ldots,a_{\ell-1})$ and $b=(b_1,\ldots,b_{\ell-1})$. Similarly we can define the preorder $\succeq$.
\end{dfn}

\begin{dfn}
\label{dfn:usual_stochastic_order}
A random vector $X\in \mathbb{N}\hspace{0.8mm}^{\ell-1}$ is said to be stochastically smaller in the usual stochastic order than a random vector  $Y \in \mathbb{N}\hspace{0.8mm}^{\ell-1}$, (denoted by $X \preceq_{\text{st}} Y$) if: $\forall \omega \in \mathbb{N}\hspace{0.8mm}^{\ell-1}$, $\Prob(X\succeq \omega)\leq \Prob(Y\succeq \omega)$.
\end{dfn}

\begin{dfn}
\label{dfn:increasing_function}
A family of random variables $\{Y_n\}_{n\in\mathbb{N}}$ is called stochastically increasing ($\preceq_{\text{st}}$-increasing) if $Y_k \preceq_{\text{st}} Y_n$ whenever $k\leq n$.
\end{dfn}
\begin{proof}[Proof of Proposition~\ref{prop:The_markov_chain}]
Markov process $\{Y_n, n\geq 1\}$,  is a multidimensional process on $E=\mathbb{N}\hspace{0.8mm}^{\ell-1}$ representing the number of innovative packets at nodes $N_1,\ldots,N_{\ell-1}$ when packet $n$ arrives at $N_1$. To prove that the Markov process $\{Y_n, n\geq 1\}$ is stochastically increasing we introduce two other processes $\{X_n, n\geq 1\}$ and $\{Z_n, n\geq 1\}$ having the same state space and transition probabilities as $\{Y_n, n\geq 1\}$.

More precisely, Markov process $\{Y_n, n\geq 1\}$ is effectively observing the evolution of the number of innovative packets present at every node of the tandem queue. We define the two new processes $\{X_n, n\geq 1\}$ and $\{Z_n, n\geq 1\}$ to observe the evolution of two other tandem queues having the same link failure probabilities as the queue of $\{Y_n, n\geq 1\}$.

\begin{figure*}[!ht]
\begin{center}
\includegraphics[width=1.0\columnwidth]{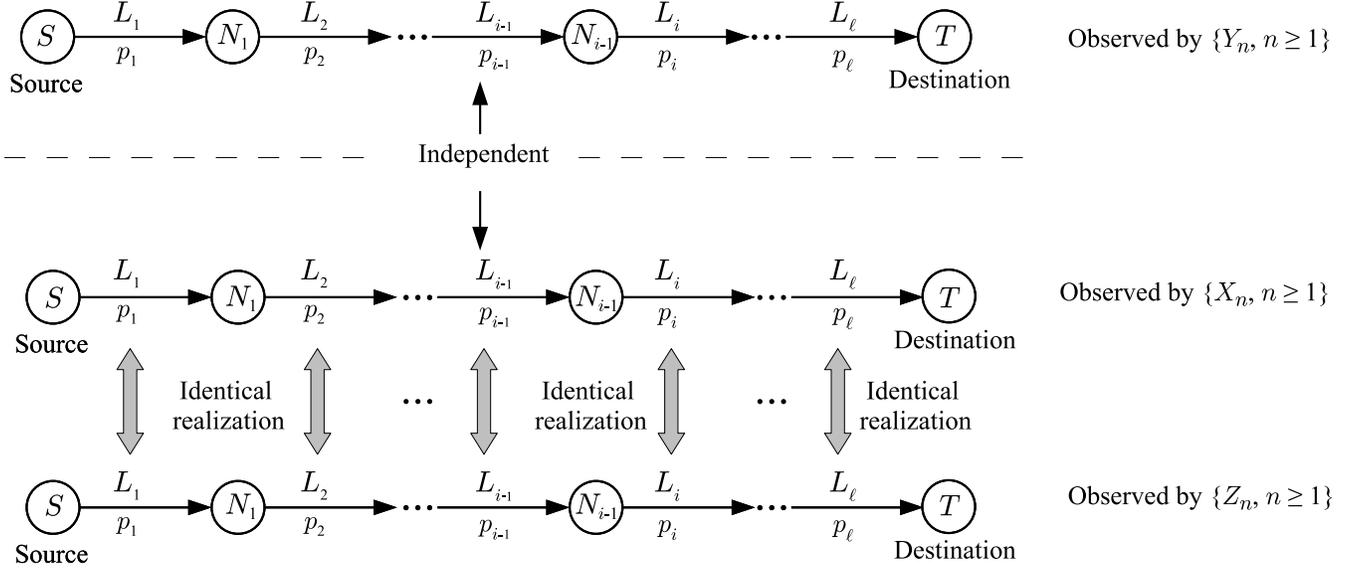}
\end{center}
\caption{Multi-hop network with the corresponding Markov chains}
\label{fig:Multi_hop_network_observed}
\end{figure*}

As seen in Figure~\ref{fig:Multi_hop_network_observed}, at each time step and at every link, the queues for $\{X_n, n\geq 1\}$ and $\{Z_n, n\geq 1\}$ either both succeed or a fail together. Moreover the successes or failures on each link on the queues observed by $\{X_n, n\geq 1\}$ and $\{Z_n, n\geq 1\}$ are independent of the successes or failures on the queue observed by $\{Y_n, n\geq 1\}$. Formally the joint process $\{(X_n, Z_n), n\geq 1\}$ constitute a coupling meaning that marginally each one of $\{X_n, n\geq 1\}$ and $\{Z_n, n\geq 1\}$ have the transition matrix $\Prob_Y$ of $\{Y_n, n\geq 1\}$. If Markov processes $\{X_n, n\geq 1\}$ and $\{Z_n, n\geq 1\}$ have different initial conditions then the following relation holds:
\begin{eqnarray}
\label{eqn:corollary_for_markov_chain}
X_1\preceq Z_1 \Rightarrow X_n\preceq Z_n
\end{eqnarray}

The proof of the above statement is very similar to the proof of Proposition 2 in \cite{aggregate}. Essentially relation (\ref{eqn:corollary_for_markov_chain}) states that since at both queues all links succeed or fail together the queue that holds more packets at each node initially ($n=1$) will also hold more packets subsequently ($n > 1$) at every node.

The initial state $Y_1$ of Markov process $\{Y_n, n\geq 1\}$ is state $\alpha = (1,0,\ldots,0)$ that is also called the minimal state since any other state is greater than the minimal state. To prove Proposition~\ref{prop:The_markov_chain} we set both processes $\{Y_n, n\geq 1\}$ and $\{X_n, n\geq 1\}$ to start from the minimal state ($Y_1\displaystyle \mathop{=}^\mathcal{D}\delta_\alpha, X_1\displaystyle \mathop{=}^\mathcal{D}\delta_\alpha\text{ where }\displaystyle\mathop{=}^\mathcal{D}$ means equality in distribution), whereas process $\{Z_n, n\geq 1\}$ has initial distribution $\mu$ that is the distribution of process $\{Y_n, n\geq 1\}$ after $(n-k)$ steps $(\mu=\Prob_Y^{n-k}\delta_\alpha\text{ and } Z_1\displaystyle \mathop{=}^\mathcal{D}\mu$). Then for every $\omega$ in the state space of $\{Y_n, n\geq 1\}$ we get:
\begin{eqnarray}
\label{eqn:lemma_condition_for_stochastic_ordering_1}
\Prob(X_n\succeq \omega)=\Prob(Y_n\succeq \omega)=\Prob(Z_k\succeq \omega)
\end{eqnarray}
where the first equality holds since the two processes have the same distribution--both start from the minimal element and have the same transition matrices--and the second equality holds since
\begin{eqnarray*}
\label{eqn:lemma_condition_for_stochastic_ordering_2}
\displaystyle Z_k\mathop{=}^\mathcal{D} \Prob_Y^k \mu\equiv \Prob_Y^k (\Prob_Y^{n-k}\delta_\alpha)=\Prob_Y^n\delta_\alpha\mathop{=}^\mathcal{D}Y_n.
\end{eqnarray*}
Moreover due to the definition of the minimal element, $X_1\preceq Z_1$ and using (\ref{eqn:corollary_for_markov_chain}) we get $X_n \preceq Z_n$. Therefore
\begin{eqnarray}
\label{eqn:lemma_condition_for_stochastic_ordering_3}
\Prob(Z_k\succeq \omega)\geq \Prob(X_k \succeq \omega)=\Prob(Y_k\succeq \omega).
\end{eqnarray}
The last equality follows from the fact that the two distributions have the same law. Equations (\ref{eqn:lemma_condition_for_stochastic_ordering_1}) and  (\ref{eqn:lemma_condition_for_stochastic_ordering_3}) conclude the proof.
\end{proof}

\section{Proof of Proposition~\ref{prop:concentration_of_Rt}}
\label{appen:concentration}

\begin{dfn}
A sequence of random variables $V_0,V_1,\ldots$ is said to be a \textbf{martingale with respect to} another sequence $U_0,U_1,\ldots$ if, for all $n\geq 0$, the following conditions hold:
\begin{itemize}
  \item $\stexp[|V_n|]<\infty $
  \item $\stexp[V_{n+1}|U_0,\ldots,U_n]=V_n$
\end{itemize}
A sequence of random variables $V_0,V_1,\dots$ is called \textbf{martingale} when it is a martingale with respect to itself. That is:
\begin{itemize}
  \item $\stexp[|V_n|]<\infty$
  \item $\stexp[V_{n+1}|V_0,...,V_n]=V_n$
\end{itemize}
\end{dfn}

\begin{theorem}
(Azuma-Hoeffding Inequality): Let $X_{0}$, $X_{1}$,...,$X_{n}$ be a martingale such that
\[B_{k}\leq X_{k}-X_{k-1} \leq B_{k}+d_{k}\]
for some constants $d_{k}$ and for some random variables $B_{k}$ that may be a function of $X_{0},...,X_{k-1}$. Then for all $t\geq 0$ and any $\lambda > 0$,
\[\Prob(|X_{t}-X_{0}|\geq \lambda)\leq 2\exp\left(-\frac{2\lambda^2}{\sum_{i=1}^{t}d_{i}^{2}}\right)\]
\label{thm:Azuma_Hoeffding_inequality}
\end{theorem}
\begin{proof}
Theorem 12.6 in \cite{mitzenmacher05probability}
\end{proof}

\begin{proof}[Proof of Proposition~\ref{prop:concentration_of_Rt}]
The proof is based on the fact that from a sequence of random variables $U_1,U_2,\ldots,U_n$ and any function $f$ it's possible to define a new sequence $V_0,\ldots,V_n$
\[\left\{
  \begin{array}{ll}
    V_0=\stexp[f(U_1,\ldots,U_n)]\\
    V_i= \stexp[f(U_1,\ldots,U_n)|U_1,\ldots,U_i]
  \end{array}
\right.\]
that is a martingale (\textit{Doob} martingale). Using the identity $\stexp[V|W]=\stexp[\stexp[V|U,W]| W]$ it's easy to verify that the above sequence $V_0,\ldots,V_n$ is indeed a martingale. Moreover if function $f$ is \textit{c-Lipschitz} and $U_1,\ldots,U_n$ are independent it can be proved that the differences $V_i-V_{i-1}$ are restricted within bounded intervals \cite{mitzenmacher05probability} (pages 305-306).

Function $R_t=g(z_{11},...,z_{tL})$ has a bounded expectation, is \textit{1-Lipschitz} and the random variables $z_{ij}$ are independent and therefore all the requirements of the above analysis hold. Specifically by setting
\begin{align*}
G_{h}=\stexp[g(z_{11},...,z_{tL})&\left|\right.\underbrace{z_{11},...,z_{kr}}]\\
                                 &\text{$h$-terms in total}
\end{align*}
we can apply the Azuma-Hoeffding inequality on the $G_{0},...,G_{tL}$ martingale and we get the following concentration result
\begin{eqnarray}
\Prob[|G_{tL}-G_{0}|\geq \lambda]= \Prob[|R_{t}-\stexp[R_{t}]|\geq \lambda] \leq 2\exp \{ -\frac{2\lambda^2}{tL} \}.
\label{eqn:intermediate_concentration_result}
\end{eqnarray}
The equality above holds since
\begin{itemize}
  \item $G_{0}    =\stexp[R_{t}]$
  \item $G_{tL}=R_{t}\text{ (the random variable itself)}$
\end{itemize}
and by substituting on (\ref{eqn:intermediate_concentration_result})  $\lambda$ with $\varepsilon_t\doteq\sqrt{\frac{tL}{2}\ell n(2t)}$
\begin{equation*}
\Prob[|R_{t}-\stexp[R_{t}]|\geq \varepsilon_t] \leq \frac{1}{t}
\end{equation*}
\end{proof}

\begin{lemma}
A legitimate choice for $t^u_n$ and $t^l_n$ is:
\begin{eqnarray*}
t^u_n=(n+n^{1/2+\delta'})/C,\text{ } \delta' \in (0,1/2)\\
t^l_n=(n-n^{1/2+\delta'})/C,\text{ } \delta' \in (0,1/2)
\end{eqnarray*}
\label{lemma:expressing_t_up_and_t_l}
\end{lemma}
\begin{IEEEproof}
For any $t \le  n/C$, the expected number of received packets $\stexp R_t$ is given by $\stexp R_t = Ct - r(t)$, where $C$ is the capacity of the network and $r(t)$ can be bounded as follows. Letting $ n_t =Ct \le  n$, we have
\begin{eqnarray*}
E(T^c_{n_t})& = &E(E(T^c_{n_t}| r(t)))\\
	& = &E(t+\bigO(r(t)))\\
	& = &t+\bigO(r(t))\end{eqnarray*}
which by Theorem~\ref{thm:super_general} implies that $r(t)$ should be $\bigO(\sqrt{n_t}) \le  \bigO(\sqrt{n})$.

The only requirement for $t^u_n$ is that it is a $t$ such that $\stexp R_t-\epsilon_t\geq n$. This is indeed true for large enough $n$ if we substitute $t^u_n$ with $(n+n^{1/2+\delta'})/C$:
\begin{align}
\stexp[R_{t^u_n}]-\epsilon_{{t^u_n}} \geq n \Rightarrow Ct^u_n-r(t^u_n)-\epsilon_{t^u_n}\geq n\Rightarrow Ct^u_n-r(t^u_n)-\sqrt{\frac{L t^u_n}{2} \ln(2t^u_n)}\geq n\notag\\
\Rightarrow C\cdot \frac{n+n^{1/2+\delta'}}{C}-r(t^u_n)- \sqrt{\frac{L(n+n^{1/2+\delta')}}{2C} \ln(\frac{2(n+n^{1/2+\delta'})}{C})}\geq n.
\label{eqn:delay_almost_last}
\end{align}

Since $r(t)\in\bigO(\sqrt{n})$ there is a constant $B>0$ such that $r(t)\leq B\sqrt{n}$ and therefore in order for (\ref{eqn:delay_almost_last}) to hold it is sufficient if
\begin{align*}
n+n^{1/2+\delta'}-B\sqrt{n}- \sqrt{\frac{L(n+n^{1/2+\delta')}}{2C} \ln(\frac{2(n+n^{1/2+\delta'})}{C})} \geq n \\
\Rightarrow n^{1/2+\delta'}\geq \sqrt{\frac{L(n+n^{1/2+\delta')}}{2C} \ln(\frac{2(n+n^{1/2+\delta'})}{C})}+B\sqrt{n}\\
\Rightarrow n^{1/2+\delta'}\geq \sqrt{n} \sqrt{\frac{L(1+n^{\delta'-1/2)}}{2C}  \ln(\frac{2(n+n^{1/2+\delta'})}{C})}+B\sqrt{n}\\
\Rightarrow  n^{\delta'} \geq \sqrt{\frac{L(1+n^{\delta'-1/2)}}{2C} \ln(\frac{2(n+n^{1/2+\delta'})}{C})}+B
\end{align*}
where the last equation holds for large enough $n$.

Similarly it can be proved that $t^l_n$ can be substituted with $(n-n^{1/2+\delta'})/C$ such that for large $n$, $\stexp R_t+\epsilon_t\leq n$.
\end{IEEEproof}
\end{appendices}

\bibliographystyle{IEEEtran}
\bibliography{IEEEabrv,NWC-abbr}
\end{document}